\documentclass[a4paper,twocolumn,english,showpacs,nofootinbib,nobibnotes,aps,pra,floatfix,superscriptaddress]{revtex4-2}
\usepackage[utf8]{inputenc}
\usepackage{amsthm}
\usepackage{amssymb}
\usepackage{graphicx}
\usepackage[caption=false]{subfig}
\usepackage{xparse}
\usepackage[table]{xcolor}
\usepackage{physics}
\usepackage{bbm} 
\usepackage{mathtools} 
\usepackage{enumitem} 
\usepackage{stackengine}
\usepackage{hyperref} 
\usepackage{algorithm}

\makeatletter
\newenvironment{protocol}
{
		\renewcommand{\ALG@name}{Protocol}
		\refstepcounter{algorithm}
        \vspace{1em}
		\hrule height.8pt depth0pt \kern2pt
		\renewcommand{\caption}[2][\relax]{
			{\raggedright\textbf{\fname@algorithm~\thealgorithm} ##2\par}%
			\ifx\relax##1\relax % 
			\addcontentsline{loa}{algorithm}{\protect\numberline{\thealgorithm}##2}%
			\else % 
			\addcontentsline{loa}{algorithm}{\protect\numberline{\thealgorithm}##1}%
			\fi
			\kern2pt\hrule\kern2pt
		}
	}{
	\kern2pt\hrule\relax
    \vspace{1em}
}
\makeatother

\usepackage{struktex}

\usepackage{times}

\makeatletter

\newtheorem{theorem}{Theorem}[section]
\newtheorem{proposition}[theorem]{Proposition}
\newtheorem{lemma}[theorem]{Lemma}
\newtheorem{corollary}[theorem]{Corollary}
\newtheorem{definition}[theorem]{Definition}

\theoremstyle{definition}
\newtheorem{examp}[theorem]{Example}

\usepackage{braket}
\usepackage{dsfont}

\usepackage[normalem]{ulem}

\newcommand{\T}{\operatorname{T}}

\begin{document}

\title{Entanglement quantification with randomized measurements is maximally difficult}

\author{Julian Eisfeld}
\email{julian.eisfeld@hhu.de}
\affiliation{Institut für Theoretische Physik III, Heinrich-Heine-Universität Düsseldorf, Universitätsstr. 1, 40225 Düsseldorf, Germany}

\author{Nikolai Wyderka}
\affiliation{Institut für Theoretische Physik III, Heinrich-Heine-Universität Düsseldorf, Universitätsstr. 1, 40225 Düsseldorf, Germany}

\date{\today}

\begin{abstract}
The certification of quantum systems is essential for emerging quantum technologies, particularly in quantum communication, networks, and distributed computing, where maintaining a common reference frame across distant nodes poses significant challenges. Reference frame independent approaches, such as randomized measurement schemes, offer a promising route by reducing experimental demands while granting access to basis-independent quantities, including entanglement. However, the efficiency of such schemes in measuring such local invariants has remained unclear. In this work, we determine the minimal number of measurement settings required to access all two-qubit invariants. We further demonstrate that entanglement certification necessarily involves the most demanding invariants, establishing it as a maximally difficult task. Our results reveal a fundamental hierarchy among invariants, with direct implications for experimental feasibility and theoretical understanding of quantum certification. Finally, we extend our analysis beyond bipartite systems by applying it to the Kempe invariant in three-qubit systems, improving known measurement protocols and providing a first step toward uncovering similar hierarchies in higher dimensions.
\end{abstract}
\maketitle
\section{Introduction}

As quantum technologies find their ways into an increasing number of applications, their precise control and certification becomes an increasingly important task that is challenged by various error sources~\cite{acin2017european, kliesch2021theory, noller2024classical}. While state preparation fidelities and the implementation of precise and robust  measurements are important for any quantum protocol, applications in quantum communication, quantum networks and distributed computing face an additional challenge, namely establishing and keeping alive a common reference frame over large distances~\cite{islam2014spatial, islam2016asynchronous,cieslinski2024analysing} which is especially problematic in satellite-based protocols~\cite{rarity2002ground,kurtsiefer2002step,liao2017satellite,li2025microsatellite}. A common way to circumvent this challenge lies in exploiting reference frame independent degrees of freedom of quantum systems, as accessing them can be achieved with significantly lower experimental demands when compared to tools like state tomography \cite{elben2020mixed, ketterer2022statistically, Wyderka}, while many important features of quantum systems, like entanglement and non-locality, are inherently basis- and therefore reference frame independent. Entanglement is of particular interest, as it acts as a resource for many communication related tasks. Despite its basis-independence, most implemented schemes to quantify it rely on basis dependent schemes like witness measurements \cite{guhne2009entanglement}.

The most popular measurement scheme that does not rely on aligned frames makes use of so-called randomized measurements \cite{Bartkiewicz, knips2020multipartite, knips2020moment, imai2021bound, elben2023randomized, cieslinski2024analysing}: by choosing the measurement directions randomly, they require only short-term stability of mutual alignments of distant parties, without being characterized in detail. Such schemes have been shown to yield access to basically all relevant basis independent quantities in two-qubit systems \cite{Wyderka} and have been used to certify even weak forms of entanglement \cite{imai2021bound}.

While Ref.~\cite{Wyderka} shows that a complete set of invariants can be measured, for some of these invariants only protocols using multiple (randomized) measurement settings were found. This significantly increases the experimental requirements, as reference frame stability has to be ensured while the measurement setting is changed, which usually takes a long time. However, it was unclear if the proposed protocols were optimal in terms of the number of required settings, or if better ones can be constructed.

In this paper, we answer this open question by exactly quantifying the required number of measurement settings in randomized measurement schemes for all two-qubit invariants, showing that the found protocols are indeed optimal.  Furthermore, we show that quantifying entanglement requires access to the most difficult invariants, implying that entanglement certification is a maximally difficult task. While this has direct practical consequences on the difficulty of implementing such schemes, the result is more fundamental: it establishes a previously unknown hierarchy among the invariants. Thus, our results directly pose related questions about similar orderings in larger-dimensional systems and consequences in related certification tasks. Finally, we take first steps in this new direction by generalizing our results to the Kempe invariant, an important invariant in three-qubit systems \cite{kempe1999multiparticle, Barnum}, and provide an improved measurement protocol for it. 

This paper is organized as follows: In Sec.~\ref{basics_ch}, we review the randomized measurement scheme and the theory of local invariants of quantum systems. Further, we formally define the type of a given invariant in terms of the number of required measurement settings, which will form the basis of our results. In Secs.~\ref{det_ch} and \ref{hodge ch}, we then calculate the types of the two-qubit invariants and establish that entanglement quantification requires access to invariants of maximal type. Finally, Sec.~\ref{kempe_ch} extends our result to the example of the Kempe invariant of three-qubit systems.

\section{Basics}\label{basics_ch}
\subsection{Randomized measurements}
We describe $n$-partite quantum states by positive-semidefinite linear maps ${\varrho : H \to H}$, such that $\tr(\varrho) = 1$, where $ H = \bigotimes\limits_{j=1}^n H_j$ is composed of the $d_j$-dimensional Hilbert spaces $H_j$ of the subsystems. Observables are described by Hermitian maps $\mathcal{O} : H\to H$. The mean value of an observable for quantum state $\varrho$ is denoted by $\langle \mathcal{O} \rangle_{\varrho} := \tr(\varrho \mathcal{O})$. To measure $\mathcal{O}$, a projective measurement in its eigenbasis is performed. Rotating the measurement basis in each subsystem $j$ can be described via unitary maps $U_j\in SU(d_j)$, leading to a rotation of the global measurement basis by $U^\dagger= U_1^\dagger \otimes ... \otimes U_n^\dagger$. This rotated measurement now results in a mean value of $\langle U^\dagger \mathcal{O} U \rangle_{\varrho} = \langle \mathcal{O}\rangle_{U\varrho U^\dagger}$. Randomly selecting such local rotations would allow to sample from the probability distribution for $\langle \mathcal{O}\rangle_{U\varrho U^\dagger}$ with $t$-th moment 
\begin{equation}\label{moments}
    \mathcal{R}^{(t)}_{\mathcal{O}} (\varrho) =\int_{SU(d_1)} ... \int_{SU(d_n)} \langle \mathcal{O} \rangle_{U\varrho U^\dagger}^t \dd \mu (U),
\end{equation}
where $\mu = \bigotimes\limits_{j=1}^n \mu_j$ and $\mu_j$ is the Haar measure (the canonical choice for random rotations) on $SU(d_j)$. By construction, the resulting quantities $\mathcal{R}^{(t)}_{\mathcal{O}} (\varrho)$ are properties of the state $\varrho$, that do not depend on a choice of measurement basis. As such, their relation to entanglement and non-locality properties of states have been the subject of previous studies \cite{imai2021bound, ketterer2022statistically, Wyderka_Geometry}.

Although technically challenging, randomly drawing measurement bases according to the Haar measure is experimentally possible \cite{wyderka2023complete, knips2020multipartite, ma2025construct}. Depending on the structure of $\mathcal{O}$, the quantity $\mathcal{R}^{(t)}_{\mathcal{O}} (\varrho)$ can therefore be obtained experimentally as follows:

If $\mathcal{O}$ is a product observable,i.e., of the form $\mathcal{O} = \mathcal{O}_1 \otimes ... \otimes \mathcal{O}_n$), then a canonical measurement protocol for $\mathcal{R}^{(t)}_{\mathcal{O}}$ is available:

\begin{protocol} \label{prot:type_1_measurement} \caption{Type 1 measurement}
\begin{enumerate}[wide, labelwidth=!, labelindent=0pt]
\item Choose a random rotation $U=U_1\otimes...\otimes U_n$.
\item Estimate the quantity $\langle U^\dagger \mathcal{O} U\rangle_\varrho$ by repeating suitably often:
\begin{enumerate}
    \item Prepare a copy of $\varrho$.
    \item In each subsystem $k$, measure the observable $U_k^\dagger \mathcal{O}_k U_k$.
\end{enumerate}
\item Estimate $\mathcal{R}^{(t)}_{\mathcal{O}}$ by repeating steps 1 and 2 suitably often.
\end{enumerate}
\end{protocol}

If, however, $\mathcal{O}$ is not a product observable, we decompose it in terms of $r\geq 2$ product observables,  ${\mathcal{O} = \sum\limits_{j=1}^r \mathcal{O}_1^{(j)} \otimes ... \otimes \mathcal{O}_n^{(j)} \equiv \sum\limits_{j=1}^r \mathcal{O}^{(j)}}$. Then, above measurement protocol is modified as follows:

\begin{protocol} \label{prot:type_r_measurement} \caption{Type $r$ measurement}
\begin{enumerate}[wide, labelwidth=!, labelindent=0pt]\item Choose a random rotation $U=U_1\otimes...\otimes U_n$.
\item[2.1.] Estimate the quantity $\langle U^\dagger \mathcal{O}^{(1)} U\rangle_\varrho$ by repeating suitably often:
\begin{enumerate}
    \item Prepare a copy of $\varrho$.
    \item In each subsystem $k$, measure the observable $U_k^\dagger \mathcal{O}^{(1)}_k U_k$.
\end{enumerate}
$\vdots$
\item[2.$r$.] Estimate the quantity $\langle U^\dagger \mathcal{O}^{(r)} U\rangle_\varrho$ by repeating suitably often:
\begin{enumerate}
    \item Prepare a copy of $\varrho$.
    \item In each subsystem $k$, measure the observable $U_k^\dagger \mathcal{O}^{(r)}_k U_k$.
\end{enumerate}
\item[3.] Estimate $\mathcal{R}^{(t)}_{\mathcal{O}}$ by repeating steps 1 and 2 suitably often.
\end{enumerate}
\end{protocol}

The measurement setup is visualized in Fig.~\ref{fig:measurement} for the example of $n=2$. Experimentally, the difficulty lies in obtaining $\langle U^\dagger\mathcal{O}U\rangle_{\varrho }$ for fixed $U=U_1 \otimes ... \otimes U_n$. This involves keeping a fairly constant reference frame for a lot of measurements. The number of measurements with a constant reference frame needed will grow by a factor of $r$ if $\mathcal{O}$ is not a product observable, and additionally, changing the measured observable is typically time consuming. Therefore, $r$ should be chosen as small as possible. The minimal $r$ for a given observable  $\mathcal{O}$ is called its \emph{tensor rank}.\footnote{Not to be confused with the number of indices of a tensor, which is sometimes also called the tensor rank.} The above described measurement procedure for an observable of tensor rank $r$ will be called \textit{type $r$ measurement} in the following.

\begin{figure}
    \centering
    \includegraphics[width=0.95\linewidth]{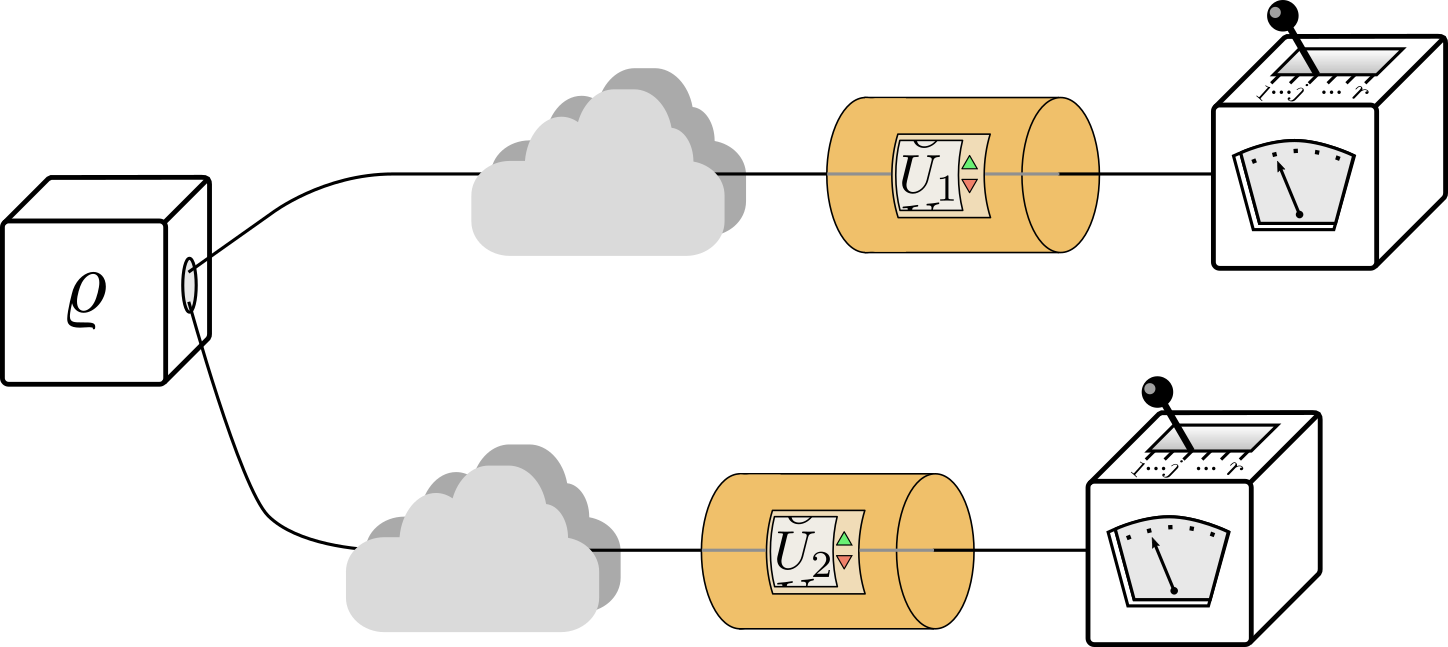}
    \caption{Visualization of the randomized type $r$ measurement protocol to measure the moments $\mathcal{R}_\mathcal{O}^{(t)}(\varrho)$ [Eq.~\eqref{moments}] for a given non-local observable $\mathcal{O}$ with tensor rank $r$. The prepared bipartite state $\varrho$ undergoes an unknown local unitary rotation, visualized by clouds. The two measurement parties then choose random and independent rotations $U_1$ and $U_2$ and cycle through $r$ different local measurement settings to measure $\langle U_1^\dagger \otimes U_2^\dagger \mathcal{O} U_1\otimes U_2 \rangle$. Repeating this procedure for different choices of $U_1$ and $U_2$ allows finally to determine the moments $\mathcal{R}_\mathcal{O}^{(t)}(\varrho)$. }
    \label{fig:measurement}
\end{figure}

In the bipartite case (i.e., $n=2$), the tensor rank equals the operator Schmidt rank, which can be calculated via the operator Schmidt decomposition and is bounded by $r_{\mathrm{max}} = \min(d_1^2,d_2^2)$ \cite{schmidt1907theorie}. Type $r_{\mathrm{max}}$ measurements will be called \textit{maximally difficult}.

Type 1 measurements are more robust against certain types of noise which can be described as a unitary change in basis over time $U_{\mathrm{noise}}(t)$. As long as the rate of change is significantly smaller than the rate at which new quantum states can be produced (which is, for example, the case when the noise takes the form of a slow drift), the noise can be considered constant during the procedure of estimating $\langle U^\dagger \mathcal{O} U\rangle_{\varrho}$ for fixed $U$ and can thus be absorbed into $U$, as it is selected at random. Changing observables takes time, such that any additional term in $\mathcal{O}$ adds on considerable measurement time and the procedure is much less robust against noise.

\subsection{Local invariants}
Any function $I$ that maps quantum states to numbers, satisfying $I(\varrho) = I(U\varrho U^\dagger)$ for all $U = U_1 \otimes ... \otimes U_n$ is called a local unitary invariant \cite{grassl1998computing}. Two quantum states $\varrho, \varrho'$ are called \textit{local unitary equivalent}, if there exists $U = U_1\otimes...\otimes U_n$, such that $\varrho = U \varrho' U^\dagger$.
It can be shown that for any $\varrho$ and $\varrho'$ which are not local unitary equivalent, there exists a polynomial local unitary invariant $I$, such that $I(\varrho) \neq I(\varrho')$ \cite{onishchik2012lie}. Moreover, the ring of polynomial local unitary invariants turns out to be finitely generated, which implies that there is a finite list of polynomials that encode all information about a quantum state, up to a local choice in basis \cite{grassl1998computing,springer2006invariant}.

Let us now focus on the two-qubit case, i.e., $n=2$ and $d_1=d_2=2$. Using the Bloch decomposition, we can write an arbitrary two-qubit state $\varrho$ as
\begin{equation}\label{density matrix main}
    \varrho = \frac{1}{4} \left[\mathbbm{1}^{\otimes 2} + \vec{\alpha} \cdot \vec{\sigma}\otimes \mathbbm{1} + \mathbbm{1} \otimes \vec{\beta} \cdot \vec{\sigma} + \sum_{j,k=1}^3 T_{jk} \sigma_j \otimes \sigma_k\right],
\end{equation}
where $\sigma_1,\sigma_2,\sigma_3$ are the Pauli matrices, $\vec{\alpha},\vec{\beta}\in \mathbb{R}^3$ are called the local Bloch vectors and $T\in \mathbb{R}^{3\times 3}$ is called the correlation matrix of $\varrho$. The Bloch vectors and correlation matrix are subject to constraints imposed by $\varrho \geq 0$.

The case of two-qubit systems is special, as it is the only system size for which a complete and finite list of polynomial invariants is known:
\begin{theorem}[Makhlin - Theorem 3 in \cite{Makhlin}]
    There are $12$ continuous and $6$ discrete local unitary invariants in the bipartite $2$-dimensional case. Writing $\varrho$ as in Eq.~\eqref{density matrix main}, the continuous invariants are given by:
    \allowdisplaybreaks
    \begin{alignat*}{3}
        I_1 &= \det(T), \quad &&I_2 = \tr(T^{\T} T), \quad &&I_3 = \tr((T^{\T} T)^2), \\
        I_4 &= |\vec{\alpha}|^2, \quad &&I_5 = |T^{\T} \vec{\alpha}|^2, \quad &&I_6 = |TT^{\T} \vec{\alpha}|^2, \\
        I_7 &= |\vec{\beta}|^2, \quad &&I_8 = |T \vec{\beta}|^2, \quad &&I_9 = |T^{\T} T \vec{\beta}|^2, \\
        I_{12} &= \langle \vec{\alpha},  T \vec{\beta} \rangle,\quad &&I_{13} = \langle \vec{\alpha}, T T^{\T} T \vec{\beta} \rangle, \quad && I_{14} = 2 \langle \vec{\alpha},\mathrm{adj}(T)\vec{\beta} \rangle.
    \end{alignat*}

    The discrete invariants are given by:
    \begin{align*}
        I_{10} &= \mathrm{sign}\left(\det(\vec{\alpha},TT^{\T}\vec{\alpha},TT^{\T} T T^{\T}  \vec{\alpha})\right),\\I_{11} &= \mathrm{sign}\left(\det(\vec{\beta},T^{\T} T \vec{\beta},T^{\T} T T^{\T}  T \vec{\beta})\right), \\
        I_{15} &= \mathrm{sign}\left(\det(\vec{\alpha},TT^{\T}  \vec{\alpha},T \vec{\beta})\right), \\
        I_{16} &= \mathrm{sign}\left(\det(T^{\T} \vec{\alpha}, \vec{\beta},T^{\T}  T \vec{\beta})\right), \\
        I_{17} &= \mathrm{sign}\left(\det(T^{\T}  \vec{\alpha},T^{\T}  T T^{\T}  \vec{\alpha},\vec{\beta})\right),\\
        I_{18} &= \mathrm{sign}\left(\det(\vec{\alpha}, T\vec{\beta},TT^{\T}  T \vec{\beta})\right).
    \end{align*}
\end{theorem}
Here, $T^{\T} $ denotes the transpose of $T$, $\mathrm{adj}(T)$ its adjugate and $\langle\vec{x}, \vec{y} \rangle$ the usual (real space) inner product of vectors $\vec{x}$ and $\vec{y}$.

While it is clear that for any observable $\mathcal{O}$ the function $\mathcal{R}^{(t)}_{\mathcal{O}}(\varrho)$ is a local unitary invariant, the explicit calculation procedure outlined in Appendix~\ref{Haar Calculation} reveals that it is indeed a polynomial invariant of degree $\leq t$ and must therefore be expressible in terms of the Makhlin invariants. Yet, finding an observable that yields information about a particular invariant is not easy. In Ref.~\cite{Wyderka}, observables $\mathcal{O}$ were found that allow measuring all continuous invariants (the discrete invariants turn out to be irrelevant for the task of entanglement quantification). For convenience, we list these observables in Tab.~\ref{obs table}. 
Notably, most of the listed observables are product observables and can therefore be measured with type $1$ measurements, with the exception of the observables for $I_1=\det(T)$ and $I_{14} = 2 \langle \vec{\alpha},\mathrm{adj}(T)\vec{\beta} \rangle$, which have operator Schmidt rank $3$ and $4$, respectively. Considering the step by step measurement procedure explained above, it is of interest to find minimum tensor rank observables for $I_1$ and $I_{14}$. To that end, we define formally what we mean by measuring an invariant.

\begin{definition}\label{def meas}
    A measurement of a local invariant $I$ refers to a series of choices of observables $(\mathcal{O}_j)_j$ and moments $(t_j)_j$ such that $I$ is uniquely determined by the quantities $\left(\mathcal{R}_{O_j}^{(t_j)}\right)_j$.
\end{definition}

In the case of the continuous polynomial invariants of a bipartite qubit system, the procedures and observables given in Ref.~\cite{Wyderka} show that any degree $\deg(I)$ invariant $I$ can be recovered from measurements with moment $t=\deg(I)$, which is the minimal required value of $t$. This is a nontrivial fact, as a priori, invariants could in principle require higher moments than $\deg(I)$ to be measured. Since higher moments $t$ lead to worse error terms in an experimental setting,
the following definitions are sensible:
\begin{definition}
~
\begin{enumerate}
    \item A local invariant of degree $\deg(I)$ is called type 1 if it can be measured through a series of 
    type 1 measurements of $\mathcal{R}^{(t_j)}_{\mathcal{O}_j}(\varrho)$ with $t_j\leq \deg(I)$.
    
     \item A local invariant of degree $\deg(I)$ is called type $\leq$ r if it can be measured through a series of 
     type $r_j\leq r$ measurements of $\mathcal{R}^{(t_j)}_{\mathcal{O}_j}(\varrho)$ with $t_j\leq \deg(I)$.
     
     \item An invariant that is of type $\leq r$ and not of type $\leq (r-1)$ is called type r. Bipartite invariants of type $r_{\mathrm{max}}=\min(d_1^2,d_2^2)$ are called maximally difficult.
\end{enumerate}
\end{definition}

Note that any function of type $\leq r$ invariants is again a type $\leq r$ invariant. This induces some additional structure, implying that the sets of all $\leq r$ invariants are non-decreasingly large subrings of the ring of all invariants.

The results of Ref.~\cite{Wyderka} imply that all continuous Makhlin invariants except of $I_1$ and $I_{14}$ are of type 1. In the following, we consider the two remaining cases individually.

\section{Determinant of correlation matrix}\label{det_ch}
Consider a two-qubit state $\varrho$ in its Bloch decomposed form of Eq.~\eqref{density matrix main}. The entries of its correlation matrix ${T \in \mathbb{R}^{3\times 3}}$ are given by $T_{ij}(\varrho) = \tr(\varrho \sigma_i \otimes \sigma_j)$. Consider now the invariant $I_1 = \det(T)$. First, note that it is impossible to measure $I_1$ using type 1 measurements, i.e., using product observables. This can be seen by noting that the Pauli matrices satisfy $\sigma_j^{\T}  = (-1)^{\delta_{j,2}} \sigma_j$, implying that  $\det(T(\varrho)) = - \det(T(\varrho^{\T_2}))$, where $\varrho^{\T_2}$ denotes the partial transposition of $\varrho$ w.r.t.~the second party. Conversely, for any product observable $\mathcal{O} = \mathcal{O}_1 \otimes \mathcal{O}_2$,
    \begin{equation*}
    \begin{split}
        \mathcal{R}^{(t)}_{\mathcal{O}} (\varrho^{\T_2}) &= \int \tr(\varrho^{\T_2} \left[(U_1 \mathcal{O}_1 U_1^\dagger) \otimes (U_2 \mathcal{O}_2 U_2^\dagger) \right])^t \dd U  \\
        &= \int \tr(\varrho \left[(U_1 \mathcal{O}_1 U_1^\dagger) \otimes (U_2^* \mathcal{O}_2^{\T}  U_2^{T}) \right])^t \dd U\\
       &= \int \tr(\varrho \left[(U_1 \mathcal{O}_1 U_1^\dagger) \otimes (U_2 \mathcal{O}_2^{\T}  U_2^\dagger) \right])^t \dd U\\
        &= \int \tr(\varrho \left[(U_1 \mathcal{O}_1 U_1^\dagger) \otimes (U_2 U' \mathcal{O}_2 U'^\dagger U_2^\dagger) \right])^t \dd U \\
        &= \mathcal{R}^{(t)}_{\mathcal{O}} (\varrho),
    \end{split}    
    \end{equation*}
since the Haar measure is invariant under complex conjugation and there exists a $U' \in SU(2)$ with $\mathcal{O}_2^{\T}  = U' \mathcal{O}_2 U'^\dagger$ leading to the substitution $U_2\mapsto U_2U'^{-1}$. This is due to the double covering of $SO(3)$ by $SU(2)$: $\mathcal{O}_2 = b_0 \mathbbm{1} + \vec{b} \cdot \vec{\sigma}$ and $\mathcal{O}^{\T} _2 = b_0 \mathbbm{1} + \vec{b}' \cdot \vec{\sigma}$ with $b_j' = (-1)^{\delta_{j,2}} b_j$ and thus $|\vec{b}|=|\vec{b}'|$.
Note that $\varrho^{\T_2}$ need not be a quantum state for $\mathcal{R}^{(t)}_{\mathcal{O}} (\varrho^{\T_2})$ to be well defined. Notably from the list of all continuous invariants in the two-qubit case, only $I_1 = \det(T)$ and $I_{14} = 2\langle \vec{\alpha}, \mathrm{adj}(T) \vec{\beta}\rangle$ are not invariant under partial transposition of $\varrho$, both gaining a minus sign instead. This establishes that the types of $I_1$ and $I_{14}$ must be larger than 1.

The following result establishes that the determinant of $T$ must be at least of type $3$:
\begin{theorem}\label{r>2 main}
    Let the moment $t\leq 3$. Then there is no bipartite observable $\mathcal{O}$ of operator Schmidt rank $r \leq 2$ such that $\mathcal{R}^{(t)}_{\mathcal{O}} (\varrho)$ depends on $\det(T)$.
\end{theorem}
\noindent This is proven as Theorem \ref{r>2} in Appendix~\ref{app:detT}.

In Ref.~\cite{Wyderka}, it was shown via explicit calculation that measurement of $\mathcal{O}_{\det} = \sum\limits_{j=1}^3 \sigma_j \otimes \sigma_j$ measures the determinant, establishing together with Theorem~\ref{r>2 main} that $I_1$ is of type $3$. However, we can show that for an arbitrary observable $\mathcal{O}$, the dependence of its randomized measurement moments $\mathcal{R}^{(3)}_{\mathcal{O}}$ on the determinant can be determined easily without need for lengthy calculations:
\begin{corollary}
    Let $\mathcal{O} = \sum\limits_{j=1}^r s_j A_j \otimes B_j$ be an arbitrary observable in its operator Schmidt decomposition, i.e., $\sum_j s_j = 1$, $s_j \geq 0$ and $\trace(A_jA_k) = \trace(B_jB_k) = \delta_{jk}$. Furthermore, let ${\tilde{P} : \mathbb{C}^{2\times2}_{\mathrm{herm.}} \to \mathbb{R}^3}$ be the projection mapping given by $A\mapsto \left(\tr(A \frac{\sigma_j}{\sqrt{2}})\right)_j$.
        
    Then
    \begin{enumerate}
        \item $\mathcal{R}^{(3)}_{\mathcal{O}} (\varrho)$ is an affine linear combination of the invariants $\{|\vec{\alpha}|^2,|\vec{\beta}|^2,\tr(T^{\T} T),\langle\vec{\alpha},T\vec{\beta}\rangle,\det(T)\}$,

        \item Let 
    \begin{align}
        M_1 := \left(\sqrt{s_1}\tilde{P}(A_1),...,\sqrt{s_r} \tilde{P}(A_r)\right) \in \mathbb{R}^{3\times r},\\
        M_2 := \left(\sqrt{s_1}\tilde{P}(B_1),...,\sqrt{s_r} \tilde{P}(B_r)\right) \in \mathbb{R}^{3\times r}.
    \end{align}
    Then the prefactor of $\det(T)$ in the affine expression of $\mathcal{R}^{(3)}_{\mathcal{O}} (\varrho)$ is given by $\det(M_1 M_2^{\T} )/8$. It is non-zero if and only if $M_1M_2^{\T} $ is invertible.
    \end{enumerate}
\end{corollary}
\noindent This is proven as Corollary \ref{class} in Appendix~\ref{app:detT}.

Under the additional restriction of a symmetric operator Schmidt decomposition ($\mathcal{O} = \sum\limits_{j=1}^r s_j A_j \otimes A_j$), all observables, whose measurement depends solely on $\det(T)$ can be classified as well:
\begin{proposition}
    Let $a,b \in \mathbb{C}, b\neq 0$. The set of observables $\mathcal{O}$ with a symmetric operator Schmidt decomposition that satisfy
    $\mathcal{R}^{(3)}_{\mathcal{O}} (\varrho)=  a + b \det(T)$ is only nonempty for $a = 0$ and $b\in \mathbb{R}_{>0}$. Its elements are precisely of the form $(U\otimes U)^\dagger \left(\sum\limits_{j=1}^3 s_j \sigma_j \otimes \sigma_j \right) (U\otimes U)$, where $U \in SU(2)$, $s_1s_2s_3 = 8b$ and $s_j>0$.
    Notably, there are no symmetric $r=4$ measurements, that only measure the determinant.
\end{proposition}
\noindent This is proven as Proposition \ref{strong char twirl} in Appendix~\ref{app:onlydetT}.

One might hope for the existence of measurements of type $\leq 2$ if we allow for higher moments exceeding three, despite the increased statistical errors. However, we can show that fourth moments yield no non-trivial dependence on the determinant:
\begin{theorem}
    There is no observable $\mathcal{O}$ with operator Schmidt rank $r\leq2$, such that $\mathcal{R}^{(4)}_{\mathcal{O}} (\varrho)$ depends on $\det(T)$.
\end{theorem}
\noindent This is proven as Theorem \ref{det t=4} in Appendix~\ref{app:t4det}.

\begin{table}[t]
    \begin{tabular}{l|l|l|c}
      \textbf{Invar.} & $t$ & \textbf{Observable} & \textbf{Type}\\
      \hline
      $I_1$ & $3$ & $\sigma_x\otimes\sigma_x+\sigma_y\otimes \sigma_y+\sigma_z\otimes \sigma_z$ & $\cellcolor{orange!50}3$\\
      $I_2$ & $2$ & $3\sigma_z \otimes \sigma_z$ & \cellcolor{green!50}$1$\\
      $I_3$ & $4$ & $\sigma_z \otimes \sigma_z$ and knowledge of $I_2$ & \cellcolor{green!50}$1$\\
      $I_4$ & $2$ & $\sqrt{3} \sigma_z\otimes \mathbbm{1}$ & \cellcolor{green!50}$1$\\
      $I_5$ & $4$ & $\sigma_z \otimes (\mathbbm{1}+\sigma_z)$ and knowledge of $I_2,I_3,I_4$ & \cellcolor{green!50}$1$\\
      $I_6$ & $6$ & $\sigma_z \otimes (\mathbbm{1}+\sigma_z)$ and knowledge of $I_2-I_5$ & \cellcolor{green!50}$1$\\
      $I_7$ & $2$ & $\sqrt{3} \mathbbm{1} \otimes \sigma_z$ & \cellcolor{green!50}$1$\\
      $I_8$ & $4$ & $(\mathbbm{1}+\sigma_z)\otimes \sigma_z$ and knowledge of $I_2,I_3,I_7$& \cellcolor{green!50}$1$\\
      $I_9$ & $6$ & $(\mathbbm{1}+\sigma_z)\otimes \sigma_z$ and knowledge of $I_2,I_3,I_7,I_8$ & \cellcolor{green!50}$1$\\
      $I_{12}$ & $3$ & $(\mathbbm{1}+\sigma_z)^{\otimes 2}$ and knowledge of $I_2,I_4,I_7$ & \cellcolor{green!50}$1$\\
      $I_{13}$ & $5$ & $(\mathbbm{1} \otimes \sigma_z)^{\otimes 2}$ and knowledge of $I_2-I_5,I_7,I_8,I_{12}$ & \cellcolor{green!50}$1$\\
      $I_{14}$ & $4$ & $\mathbbm{1}\otimes \sigma_x + \sigma_x \otimes \mathbbm{1} + \sigma_y\otimes \sigma_z \pm \sigma_z \otimes \sigma_y$ & \cellcolor{red!50}$4$\\
    \end{tabular}

    \caption{List of continuous invariants by type with optimal observables. If a single observable is listed, then $\mathcal{R}^{(t)}_{\mathcal{O}} = I_j$. If an observable is listed alongside a list of other Invariants, then $\mathcal{R}^{(t)}_{\mathcal{O}} = aI_j+b$, where $a \in \mathbb{R} \backslash \{0\}$ and $b$ only depends on the listed known invariants. $I_{14}$ is given (up to a prefactor) as the difference between the measurements of the two listed observables. These observables were first found by Wyderka et al. in \cite{Wyderka}. This is also where the prefactors can be found, or at least (in the case of $I_5,I_6,I_8$ and $I_9$) can be derived from.}\label{obs table}
\end{table}

\section{Hodge invariant}\label{hodge ch}
The Hodge invariant\footnote{The name was chosen in Ref.~\cite{Wyderka} due to the fact that it can be expressed using the Hodge star operator} $I_{14} = 2\langle \vec{\alpha}, \mathrm{adj}(T)\vec{\beta} \rangle$ is the only other continuous Makhlin invariant that cannot be of type 1. In Ref.~\cite{Wyderka}, it was shown that it is possible to reconstruct it from a type $r=4$ measurement, which is the worst possible option for an invariant of a bipartite qubit system from the standpoint of measurement difficulty.
There, the observable
\begin{equation}
    \mathcal{O}_{\text{Hodge}} = \mathbbm{1} \otimes \sigma_1 + \sigma_1 \otimes \mathbbm{1} + \sigma_2 \otimes \sigma_3 + \sigma_3 \otimes \sigma_2
\end{equation}
was used. 

Using the same general proof strategy as for Theorem~\ref{r>2 main}, we can show that $I_{14}$ indeed requires type $4$ measurements and no type $\leq 3$ measurement can be found:
\begin{theorem}
    $\mathcal{R}^{(4)}_{\mathcal{O}}(\varrho)$ does not depend on $I_{14}$ if the operator Schmidt rank $r$ of $\mathcal{O}$ satisfies $r\leq 3$. Thus, the Hodge invariant is of type 4 and is therefore maximally difficult.
\end{theorem}
\noindent This is proven as Theorem \ref{hodge imp} in Appendix~\ref{app:t4hodge}.

This result completes the hierarchy of invariants listed in Tab.~\ref{obs table}. Let us conclude by linking the results to the problem of entanglement quantification. It is known that knowledge of all continuous Makhlin invariants allows to determine the negativity $N(\varrho) = (\tr(|\varrho^{\T_2}|)-1)/2$ of a two-qubit quantum state $\varrho$, which is a faithful entanglement measure \cite{vidal2002computable, Bartkiewicz}. Notably, it explicitly involves both the determinant and the Hodge invariant. Together with the fact that all entanglement measures of two-qubit systems are equivalent in the sense that the induced orderings of quantum states are compatible, this implies: 
\begin{theorem}
    Evaluating entanglement measures in a two-qubit system via randomized measurements is maximally difficult, i.e., requires type $4$ measurements.
\end{theorem}

\section{Kempe invariant}\label{kempe_ch}
Let us generalize some of our results to the three-partite scenario. Analogous to the two-qubit case, we can decompose a three-qubit state in terms of Pauli matrices. A generic state $\varrho$ can then be written as
\begin{equation}\label{bloch 3}
\begin{split}
    \varrho = \frac{1}{8} &\big[\mathbbm{1}^{\otimes 3}+\vec{\alpha}\!\cdot\!\vec{\sigma}\!\otimes\! \mathbbm{1}\!\otimes\! \mathbbm{1} + \mathbbm{1}\!\otimes\!\vec{\beta} \!\cdot\!  \vec{\sigma}\!\otimes \!\mathbbm{1}+ \mathbbm{1}\!\otimes\!\mathbbm{1}\!\otimes\!\vec{\gamma}\!\cdot\!\vec{\sigma} \\
    &+\sum_{j,k=1}^3 \big(T_{jk}^{AB} \sigma_j \!\otimes\! \sigma_k \!\otimes\! \mathbbm{1} + T_{jk}^{BC} \mathbbm{1} \!\otimes\! \sigma_j\!\otimes\! \sigma_k \\
    &+ T^{CA}_{jk} \sigma_k \!\otimes\! \mathbbm{1}\!\otimes\! \sigma_j \big) 
    + \!\!\!\sum_{j,k,\ell =1}^3 W_{jk\ell} \sigma_j \!\otimes\! \sigma_k \!\otimes\! \sigma_\ell\big].
\end{split}
\end{equation}

In contrast to the two-qubit case, there is no complete set of local unitary invariants of three-qubit states known. However, some interesting invariants are known, including the \textit{Kempe invariant}, first introduced in Ref.~\cite{kempe1999multiparticle} and later generalized to mixed states in Ref.~\cite{Barnum}. In terms of above decomposition, it can be written as \cite{Wyderka}
\begin{equation*}
\begin{split}
    I_{\mathrm{Kempe}} =& \frac{1}{8} \big[1+|\vec{\alpha}|^2 + |\vec{\beta}|^2 + |\vec{\gamma}|^2 + \langle \vec{\alpha}, T^{AB} \vec{\beta} \rangle \\
    &+ \langle \vec{\beta}, T^{BC} \vec{\gamma} \rangle + \langle \vec{\gamma}, T^{CA} \vec{\alpha} \rangle + \tr(T^{AB}T^{BC}T^{CA})\big],
\end{split}
\end{equation*}
Notice that each of the first seven terms is a type $1$ invariants of a $2$-qubit subsystem and can thus be easily obtained using observables of the form $\mathcal{O}_1 \otimes \mathcal {O}_2\otimes \mathbbm{1}$, $\mathcal{O}_1 \otimes  \mathbbm{1} \otimes \mathcal {O}_3$ and ${\mathbbm{1}\otimes\mathcal{O}_2 \otimes \mathcal {O}_3}$. Only the last term, $\tr(T^{AB}T^{BC}T^{CA})$, requires further classification.

Notably, partial transposition of $\varrho$ w.r.t.~any of the subsystems induces minus signs in two of the three correlation matrices, leaving the invariant unchanged. 
Comparing to the bipartite case, one might think that the only obstruction of measuring invariants using product observables is the non-invariance under partial transposition, implying that there should be a type $1$ measurement procedure for $\tr(T^{AB}T^{BC}T^{CA})$. However, this is just partially true, as it can only be measured in a linear combination with other invariants:
\begin{proposition}
~
\begin{enumerate}
    \item There is no product observable $\mathcal{O} = A \otimes B \otimes C$ such that the invariant $\tr(T^{AB}T^{BC}T^{CA})$ is uniquely recoverable from randomized measurement moments of order up to 3.

    \item There exists a product observable such that
    \begin{equation}\label{r=1 Kempe Eq}
        \begin{split}
            3 ||W||_2^2 &+ 6 \sum_{jk\ell} W_{jk\ell} \left[ T^{AB}_{jk} \gamma_\ell + \alpha_j T^{BC}_{k\ell} + T^{CA}_{\ell j} \beta_k\right] \\
        &+ 6 \tr(T^{AB}T^{BC}T^{CA}).
        \end{split}
    \end{equation}
    can be obtained from a $t=3$ randomized measurement moment.
\end{enumerate}
\end{proposition}
\noindent This is proven as Proposition \ref{Kempe Prop} in the Appendix.

Due to the large amount of symmetry in the explicit calculations for  $r=1$, the five different invariants that appear in Eq.~\eqref{r=1 Kempe Eq} only ever appear together. They can only be decoupled using type $2$ measurements:
\begin{proposition}
    The invariant $\tr(T^{AB}T^{BC}T^{CA})$ (and therefore $I_{\mathrm{Kempe}}$) can be uniquely recovered by measuring the randomized measurement moments of the tensor rank $2$ observables
    \begin{align}\label{Kempe obs main}
    \mathcal{O}_1&=\mathbbm{1}\otimes \mathbbm{1}\otimes \mathbbm{1} + \sigma_3 \otimes \sigma_3 \otimes \sigma_3,\\
    \mathcal{O}_2&=\mathbbm{1} \otimes \mathbbm{1} \otimes \left(\mathbbm{1} -\sigma_1 \right)+ \sigma_3\otimes \sigma_3 \otimes \left(\mathbbm{1} -\sigma_1 \right),\\
    \mathcal{O}_3&=\mathbbm{1} \otimes \left(\mathbbm{1} -\sigma_1 \right) \otimes \mathbbm{1}+ \sigma_3 \otimes \left(\mathbbm{1} -\sigma_1 \right) \otimes \sigma_3,\\
    \mathcal{O}_4&=\left(\mathbbm{1} -\sigma_1 \right) \otimes \mathbbm{1}\otimes \mathbbm{1}+ \left(\mathbbm{1} -\sigma_1 \right) \otimes \sigma_3 \otimes \sigma_3
    \end{align}
in addition to some rank $1$ observables. The first observable allows measurement of $||W||_2^2$.
\end{proposition}
\noindent It is proven as Proposition \ref{kempe prop 2} in Appendix~\ref{app:kempe}.

Notably, this result implies that the Kempe invariant is accessible through a type $2$ measurement. This improves the previously known result from Ref.~\cite{Wyderka}, where only a type $3$ measurement was found.
Further note that the first observable in Eq.~\eqref{Kempe obs main} is effectively a type 1 measurement, since a measurement of $\mathbbm{1}^{\otimes 3}$ is independent of the quantum state.

\section{Discussion}\label{discussion_sec}
In this paper, we fully classified all continuous local unitary invariants of bipartite qubit systems w.r.t.~the difficulty of obtaining their value using randomized measurements, which lead to the results in Table \ref{obs table}. Complete characterization was possible since a complete set of invariants, determining any quantum state up to local unitary transformations, is known. We established that previously known measurement schemes for the invariants are indeed optimal. 

While in the case of two-qubit systems, the invariance under partial transposition proved to be decisive on whether a given invariant can be trivially measured with type $1$ measurements, we showed that this criterion fails in three-partite systems. There, we established that the Kempe invariant requires at least type $2$ measurements and show the optimality of this bound by providing an explicit protocol, thereby improving previously known results which required type $3$. 

Building on these results, a couple of further questions and possible generalizations arise.

First of all, this paper did not consider the discrete invariants from Ref.~\cite{Makhlin}. It remains open whether they can be measured by randomized measurements and of what type their measurement is. The fact that they are all observables of degree six or higher enforces forces $t \geq 6$. This makes the approach from this paper infeasible, since this would mean generating a $6! \times 6! = 720 \times 720$-dimensional matrix and analyzing its likely massive nullspace. For $I_{10}$ and $I_{11}$, a $9! \times 9!$ matrix would be required, which has around 131 billion entries.

Furthermore, we restricted ourselves to measurement moments of order $t \leq 4$, and we cannot rule out the existence of a type $2$ measurement of $I_1$ or $I_{14}$ to exist using  higher $t$ with the tools developed in this paper. Proving or disproving its existence is an interesting open question for further research, although such measurement would suffer from high statistical errors, making their advantage over type $3$ measurements questionable.

Another peculiar observation concerning the measurement of the Kempe invariant raises questions of a different kind: we found the invariant $||W||_2^2$ to be measurable using the observable $\mathcal{O}_1$ in Eq.~\eqref{Kempe obs main}. It is of type $2$, but effectively only requires type 1 measurements, since the result of the measurement $\lambda \mathbbm{1}\otimes \mathbbm{1}\otimes \mathbbm{1}$ is always deterministically given by $\lambda$ due to the trace of the state being fixed. Taking this into account could lead to a more precise classification of invariants up to trivial measurements. It is already known from the explicit characterization in Sec.~\ref{det_ch} that the determinant remains to be of type 3 even when accounting for additional trivial measurements (since $\tilde{P}(\lambda \mathbbm{1}) = 0$). A similar result is not yet known for the Hodge or Kempe invariant, but should be derivable using the methods from this paper.

Finally, we note that most of the formalism of this paper can be adapted to systems of local dimension $d \geq 3$ and could be used to extend our results to larger system sizes.

\bibliographystyle{apsrev4-1}
\bibliography{cite}

\begin{thebibliography}{34}%
\makeatletter
\providecommand \@ifxundefined [1]{%
 \@ifx{#1\undefined}
}%
\providecommand \@ifnum [1]{%
 \ifnum #1\expandafter \@firstoftwo
 \else \expandafter \@secondoftwo
 \fi
}%
\providecommand \@ifx [1]{%
 \ifx #1\expandafter \@firstoftwo
 \else \expandafter \@secondoftwo
 \fi
}%
\providecommand \natexlab [1]{#1}%
\providecommand \enquote  [1]{``#1''}%
\providecommand \bibnamefont  [1]{#1}%
\providecommand \bibfnamefont [1]{#1}%
\providecommand \citenamefont [1]{#1}%
\providecommand \href@noop [0]{\@secondoftwo}%
\providecommand \href [0]{\begingroup \@sanitize@url \@href}%
\providecommand \@href[1]{\@@startlink{#1}\@@href}%
\providecommand \@@href[1]{\endgroup#1\@@endlink}%
\providecommand \@sanitize@url [0]{\catcode `\\12\catcode `\$12\catcode
  `\&12\catcode `\#12\catcode `\^12\catcode `\_12\catcode `\%12\relax}%
\providecommand \@@startlink[1]{}%
\providecommand \@@endlink[0]{}%
\providecommand \url  [0]{\begingroup\@sanitize@url \@url }%
\providecommand \@url [1]{\endgroup\@href {#1}{\urlprefix }}%
\providecommand \urlprefix  [0]{URL }%
\providecommand \Eprint [0]{\href }%
\providecommand \doibase [0]{http://dx.doi.org/}%
\providecommand \selectlanguage [0]{\@gobble}%
\providecommand \bibinfo  [0]{\@secondoftwo}%
\providecommand \bibfield  [0]{\@secondoftwo}%
\providecommand \translation [1]{[#1]}%
\providecommand \BibitemOpen [0]{}%
\providecommand \bibitemStop [0]{}%
\providecommand \bibitemNoStop [0]{.\EOS\space}%
\providecommand \EOS [0]{\spacefactor3000\relax}%
\providecommand \BibitemShut  [1]{\csname bibitem#1\endcsname}%
\let\auto@bib@innerbib\@empty
%</preamble>
\bibitem [{\citenamefont {Ac{\'\i}n}\ \emph {et~al.}(2018)\citenamefont
  {Ac{\'\i}n}, \citenamefont {Bloch}, \citenamefont {Buhrman}, \citenamefont
  {Calarco}, \citenamefont {Eichler}, \citenamefont {Eisert}, \citenamefont
  {Esteve}, \citenamefont {Gisin}, \citenamefont {Glaser}, \citenamefont
  {Jelezko} \emph {et~al.}}]{acin2017european}%
  \BibitemOpen
  \bibfield  {author} {\bibinfo {author} {\bibfnamefont {A.}~\bibnamefont
  {Ac{\'\i}n}}, \bibinfo {author} {\bibfnamefont {I.}~\bibnamefont {Bloch}},
  \bibinfo {author} {\bibfnamefont {H.}~\bibnamefont {Buhrman}}, \bibinfo
  {author} {\bibfnamefont {T.}~\bibnamefont {Calarco}}, \bibinfo {author}
  {\bibfnamefont {C.}~\bibnamefont {Eichler}}, \bibinfo {author} {\bibfnamefont
  {J.}~\bibnamefont {Eisert}}, \bibinfo {author} {\bibfnamefont
  {D.}~\bibnamefont {Esteve}}, \bibinfo {author} {\bibfnamefont
  {N.}~\bibnamefont {Gisin}}, \bibinfo {author} {\bibfnamefont {S.~J.}\
  \bibnamefont {Glaser}}, \bibinfo {author} {\bibfnamefont {F.}~\bibnamefont
  {Jelezko}},  \emph {et~al.},\ }\href {\doibase 10.1088/1367-2630/aad1ea}
  {\bibfield  {journal} {\bibinfo  {journal} {New J. Phys.}\ }\textbf {\bibinfo
  {volume} {20}},\ \bibinfo {pages} {080201} (\bibinfo {year}
  {2018})}\BibitemShut {NoStop}%
\bibitem [{\citenamefont {Kliesch}\ and\ \citenamefont
  {Roth}(2021)}]{kliesch2021theory}%
  \BibitemOpen
  \bibfield  {author} {\bibinfo {author} {\bibfnamefont {M.}~\bibnamefont
  {Kliesch}}\ and\ \bibinfo {author} {\bibfnamefont {I.}~\bibnamefont {Roth}},\
  }\href {\doibase 10.1103/PRXQuantum.2.010201} {\bibfield  {journal} {\bibinfo
   {journal} {PRX Quantum}\ }\textbf {\bibinfo {volume} {2}},\ \bibinfo {pages}
  {010201} (\bibinfo {year} {2021})}\BibitemShut {NoStop}%
\bibitem [{\citenamefont {N{\"o}ller}\ \emph {et~al.}(2025)\citenamefont
  {N{\"o}ller}, \citenamefont {Miklin}, \citenamefont {Kliesch},\ and\
  \citenamefont {Gachechiladze}}]{noller2024classical}%
  \BibitemOpen
  \bibfield  {author} {\bibinfo {author} {\bibfnamefont {J.}~\bibnamefont
  {N{\"o}ller}}, \bibinfo {author} {\bibfnamefont {N.}~\bibnamefont {Miklin}},
  \bibinfo {author} {\bibfnamefont {M.}~\bibnamefont {Kliesch}}, \ and\
  \bibinfo {author} {\bibfnamefont {M.}~\bibnamefont {Gachechiladze}},\ }\href
  {\doibase 10.22331/q-2025-08-08-1825} {\bibfield  {journal} {\bibinfo
  {journal} {Quantum}\ }\textbf {\bibinfo {volume} {9}},\ \bibinfo {pages}
  {1825} (\bibinfo {year} {2025})}\BibitemShut {NoStop}%
\bibitem [{\citenamefont {Islam}\ \emph {et~al.}(2014)\citenamefont {Islam},
  \citenamefont {Magnin}, \citenamefont {Sorg},\ and\ \citenamefont
  {Wehner}}]{islam2014spatial}%
  \BibitemOpen
  \bibfield  {author} {\bibinfo {author} {\bibfnamefont {T.}~\bibnamefont
  {Islam}}, \bibinfo {author} {\bibfnamefont {L.}~\bibnamefont {Magnin}},
  \bibinfo {author} {\bibfnamefont {B.}~\bibnamefont {Sorg}}, \ and\ \bibinfo
  {author} {\bibfnamefont {S.}~\bibnamefont {Wehner}},\ }\href {\doibase
  10.1088/1367-2630/16/6/063040} {\bibfield  {journal} {\bibinfo  {journal}
  {New J. Phys.}\ }\textbf {\bibinfo {volume} {16}},\ \bibinfo {pages} {063040}
  (\bibinfo {year} {2014})}\BibitemShut {NoStop}%
\bibitem [{\citenamefont {Islam}\ and\ \citenamefont
  {Wehner}(2016)}]{islam2016asynchronous}%
  \BibitemOpen
  \bibfield  {author} {\bibinfo {author} {\bibfnamefont {T.}~\bibnamefont
  {Islam}}\ and\ \bibinfo {author} {\bibfnamefont {S.}~\bibnamefont {Wehner}},\
  }\href {\doibase 10.1088/1367-2630/18/3/033018} {\bibfield  {journal}
  {\bibinfo  {journal} {New J. Phys.}\ }\textbf {\bibinfo {volume} {18}},\
  \bibinfo {pages} {033018} (\bibinfo {year} {2016})}\BibitemShut {NoStop}%
\bibitem [{\citenamefont {Cie{\'s}li{\'n}ski}\ \emph
  {et~al.}(2024)\citenamefont {Cie{\'s}li{\'n}ski}, \citenamefont {Imai},
  \citenamefont {Dziewior}, \citenamefont {G{\"u}hne}, \citenamefont {Knips},
  \citenamefont {Laskowski}, \citenamefont {Meinecke}, \citenamefont
  {Paterek},\ and\ \citenamefont {V{\'e}rtesi}}]{cieslinski2024analysing}%
  \BibitemOpen
  \bibfield  {author} {\bibinfo {author} {\bibfnamefont {P.}~\bibnamefont
  {Cie{\'s}li{\'n}ski}}, \bibinfo {author} {\bibfnamefont {S.}~\bibnamefont
  {Imai}}, \bibinfo {author} {\bibfnamefont {J.}~\bibnamefont {Dziewior}},
  \bibinfo {author} {\bibfnamefont {O.}~\bibnamefont {G{\"u}hne}}, \bibinfo
  {author} {\bibfnamefont {L.}~\bibnamefont {Knips}}, \bibinfo {author}
  {\bibfnamefont {W.}~\bibnamefont {Laskowski}}, \bibinfo {author}
  {\bibfnamefont {J.}~\bibnamefont {Meinecke}}, \bibinfo {author}
  {\bibfnamefont {T.}~\bibnamefont {Paterek}}, \ and\ \bibinfo {author}
  {\bibfnamefont {T.}~\bibnamefont {V{\'e}rtesi}},\ }\href {\doibase
  10.1016/j.physrep.2024.09.009} {\bibfield  {journal} {\bibinfo  {journal}
  {Phys. Rep.}\ }\textbf {\bibinfo {volume} {1095}},\ \bibinfo {pages} {1}
  (\bibinfo {year} {2024})}\BibitemShut {NoStop}%
\bibitem [{\citenamefont {Rarity}\ \emph {et~al.}(2002)\citenamefont {Rarity},
  \citenamefont {Tapster}, \citenamefont {Gorman},\ and\ \citenamefont
  {Knight}}]{rarity2002ground}%
  \BibitemOpen
  \bibfield  {author} {\bibinfo {author} {\bibfnamefont {J.~G.}\ \bibnamefont
  {Rarity}}, \bibinfo {author} {\bibfnamefont {P.}~\bibnamefont {Tapster}},
  \bibinfo {author} {\bibfnamefont {P.}~\bibnamefont {Gorman}}, \ and\ \bibinfo
  {author} {\bibfnamefont {P.}~\bibnamefont {Knight}},\ }\href {\doibase
  10.1088/1367-2630/4/1/382} {\bibfield  {journal} {\bibinfo  {journal} {New J.
  Phys.}\ }\textbf {\bibinfo {volume} {4}},\ \bibinfo {pages} {82} (\bibinfo
  {year} {2002})}\BibitemShut {NoStop}%
\bibitem [{\citenamefont {Kurtsiefer}\ \emph {et~al.}(2002)\citenamefont
  {Kurtsiefer}, \citenamefont {Zarda}, \citenamefont {Halder}, \citenamefont
  {Weinfurter}, \citenamefont {Gorman}, \citenamefont {Tapster},\ and\
  \citenamefont {Rarity}}]{kurtsiefer2002step}%
  \BibitemOpen
  \bibfield  {author} {\bibinfo {author} {\bibfnamefont {C.}~\bibnamefont
  {Kurtsiefer}}, \bibinfo {author} {\bibfnamefont {P.}~\bibnamefont {Zarda}},
  \bibinfo {author} {\bibfnamefont {M.}~\bibnamefont {Halder}}, \bibinfo
  {author} {\bibfnamefont {H.}~\bibnamefont {Weinfurter}}, \bibinfo {author}
  {\bibfnamefont {P.}~\bibnamefont {Gorman}}, \bibinfo {author} {\bibfnamefont
  {P.}~\bibnamefont {Tapster}}, \ and\ \bibinfo {author} {\bibfnamefont
  {J.}~\bibnamefont {Rarity}},\ }\href {\doibase 10.1038/419450a} {\bibfield
  {journal} {\bibinfo  {journal} {Nature}\ }\textbf {\bibinfo {volume} {419}},\
  \bibinfo {pages} {450} (\bibinfo {year} {2002})}\BibitemShut {NoStop}%
\bibitem [{\citenamefont {Liao}\ \emph {et~al.}(2017)\citenamefont {Liao},
  \citenamefont {Cai}, \citenamefont {Liu}, \citenamefont {Zhang},
  \citenamefont {Li}, \citenamefont {Ren}, \citenamefont {Yin}, \citenamefont
  {Shen}, \citenamefont {Cao}, \citenamefont {Li} \emph
  {et~al.}}]{liao2017satellite}%
  \BibitemOpen
  \bibfield  {author} {\bibinfo {author} {\bibfnamefont {S.-K.}\ \bibnamefont
  {Liao}}, \bibinfo {author} {\bibfnamefont {W.-Q.}\ \bibnamefont {Cai}},
  \bibinfo {author} {\bibfnamefont {W.-Y.}\ \bibnamefont {Liu}}, \bibinfo
  {author} {\bibfnamefont {L.}~\bibnamefont {Zhang}}, \bibinfo {author}
  {\bibfnamefont {Y.}~\bibnamefont {Li}}, \bibinfo {author} {\bibfnamefont
  {J.-G.}\ \bibnamefont {Ren}}, \bibinfo {author} {\bibfnamefont
  {J.}~\bibnamefont {Yin}}, \bibinfo {author} {\bibfnamefont {Q.}~\bibnamefont
  {Shen}}, \bibinfo {author} {\bibfnamefont {Y.}~\bibnamefont {Cao}}, \bibinfo
  {author} {\bibfnamefont {Z.-P.}\ \bibnamefont {Li}},  \emph {et~al.},\ }\href
  {\doibase 10.1038/nature23655} {\bibfield  {journal} {\bibinfo  {journal}
  {Nature}\ }\textbf {\bibinfo {volume} {549}},\ \bibinfo {pages} {43}
  (\bibinfo {year} {2017})}\BibitemShut {NoStop}%
\bibitem [{\citenamefont {Li}\ \emph {et~al.}(2025)\citenamefont {Li},
  \citenamefont {Cai}, \citenamefont {Ren}, \citenamefont {Wang}, \citenamefont
  {Yang}, \citenamefont {Zhang}, \citenamefont {Wu}, \citenamefont {Chang},
  \citenamefont {Wu}, \citenamefont {Jin} \emph
  {et~al.}}]{li2025microsatellite}%
  \BibitemOpen
  \bibfield  {author} {\bibinfo {author} {\bibfnamefont {Y.}~\bibnamefont
  {Li}}, \bibinfo {author} {\bibfnamefont {W.-Q.}\ \bibnamefont {Cai}},
  \bibinfo {author} {\bibfnamefont {J.-G.}\ \bibnamefont {Ren}}, \bibinfo
  {author} {\bibfnamefont {C.-Z.}\ \bibnamefont {Wang}}, \bibinfo {author}
  {\bibfnamefont {M.}~\bibnamefont {Yang}}, \bibinfo {author} {\bibfnamefont
  {L.}~\bibnamefont {Zhang}}, \bibinfo {author} {\bibfnamefont {H.-Y.}\
  \bibnamefont {Wu}}, \bibinfo {author} {\bibfnamefont {L.}~\bibnamefont
  {Chang}}, \bibinfo {author} {\bibfnamefont {J.-C.}\ \bibnamefont {Wu}},
  \bibinfo {author} {\bibfnamefont {B.}~\bibnamefont {Jin}},  \emph {et~al.},\
  }\href {\doibase 10.1038/s41586-025-08739-z} {\bibfield  {journal} {\bibinfo
  {journal} {Nature}\ }\textbf {\bibinfo {volume} {640}},\ \bibinfo {pages}
  {47} (\bibinfo {year} {2025})}\BibitemShut {NoStop}%
\bibitem [{\citenamefont {Elben}\ \emph {et~al.}(2020)\citenamefont {Elben},
  \citenamefont {Kueng}, \citenamefont {Huang}, \citenamefont {van Bijnen},
  \citenamefont {Kokail}, \citenamefont {Dalmonte}, \citenamefont {Calabrese},
  \citenamefont {Kraus}, \citenamefont {Preskill}, \citenamefont {Zoller} \emph
  {et~al.}}]{elben2020mixed}%
  \BibitemOpen
  \bibfield  {author} {\bibinfo {author} {\bibfnamefont {A.}~\bibnamefont
  {Elben}}, \bibinfo {author} {\bibfnamefont {R.}~\bibnamefont {Kueng}},
  \bibinfo {author} {\bibfnamefont {H.-Y.}\ \bibnamefont {Huang}}, \bibinfo
  {author} {\bibfnamefont {R.}~\bibnamefont {van Bijnen}}, \bibinfo {author}
  {\bibfnamefont {C.}~\bibnamefont {Kokail}}, \bibinfo {author} {\bibfnamefont
  {M.}~\bibnamefont {Dalmonte}}, \bibinfo {author} {\bibfnamefont
  {P.}~\bibnamefont {Calabrese}}, \bibinfo {author} {\bibfnamefont
  {B.}~\bibnamefont {Kraus}}, \bibinfo {author} {\bibfnamefont
  {J.}~\bibnamefont {Preskill}}, \bibinfo {author} {\bibfnamefont
  {P.}~\bibnamefont {Zoller}},  \emph {et~al.},\ }\href {\doibase
  10.1103/PhysRevLett.125.200501} {\bibfield  {journal} {\bibinfo  {journal}
  {Phys. Rev. Lett.}\ }\textbf {\bibinfo {volume} {125}},\ \bibinfo {pages}
  {200501} (\bibinfo {year} {2020})}\BibitemShut {NoStop}%
\bibitem [{\citenamefont {Ketterer}\ \emph {et~al.}(2022)\citenamefont
  {Ketterer}, \citenamefont {Imai}, \citenamefont {Wyderka},\ and\
  \citenamefont {G{\"u}hne}}]{ketterer2022statistically}%
  \BibitemOpen
  \bibfield  {author} {\bibinfo {author} {\bibfnamefont {A.}~\bibnamefont
  {Ketterer}}, \bibinfo {author} {\bibfnamefont {S.}~\bibnamefont {Imai}},
  \bibinfo {author} {\bibfnamefont {N.}~\bibnamefont {Wyderka}}, \ and\
  \bibinfo {author} {\bibfnamefont {O.}~\bibnamefont {G{\"u}hne}},\ }\href
  {\doibase 10.1103/PhysRevA.106.L010402} {\bibfield  {journal} {\bibinfo
  {journal} {Phys. Rev. A}\ }\textbf {\bibinfo {volume} {106}},\ \bibinfo
  {pages} {L010402} (\bibinfo {year} {2022})}\BibitemShut {NoStop}%
\bibitem [{\citenamefont {Wyderka}\ \emph
  {et~al.}(2023{\natexlab{a}})\citenamefont {Wyderka}, \citenamefont
  {Ketterer}, \citenamefont {Imai}, \citenamefont {Bönsel}, \citenamefont
  {Jones}, \citenamefont {Kirby}, \citenamefont {Yu},\ and\ \citenamefont
  {Gühne}}]{Wyderka}%
  \BibitemOpen
  \bibfield  {author} {\bibinfo {author} {\bibfnamefont {N.}~\bibnamefont
  {Wyderka}}, \bibinfo {author} {\bibfnamefont {A.}~\bibnamefont {Ketterer}},
  \bibinfo {author} {\bibfnamefont {S.}~\bibnamefont {Imai}}, \bibinfo {author}
  {\bibfnamefont {J.~L.}\ \bibnamefont {Bönsel}}, \bibinfo {author}
  {\bibfnamefont {D.~E.}\ \bibnamefont {Jones}}, \bibinfo {author}
  {\bibfnamefont {B.~T.}\ \bibnamefont {Kirby}}, \bibinfo {author}
  {\bibfnamefont {X.-D.}\ \bibnamefont {Yu}}, \ and\ \bibinfo {author}
  {\bibfnamefont {O.}~\bibnamefont {Gühne}},\ }\href {\doibase
  10.1103/physrevlett.131.090201} {\bibfield  {journal} {\bibinfo  {journal}
  {Phys. Rev. Lett.}\ }\textbf {\bibinfo {volume} {131}},\ \bibinfo {pages}
  {090201} (\bibinfo {year} {2023}{\natexlab{a}})}\BibitemShut {NoStop}%
\bibitem [{\citenamefont {G{\"u}hne}\ and\ \citenamefont
  {T{\'o}th}(2009)}]{guhne2009entanglement}%
  \BibitemOpen
  \bibfield  {author} {\bibinfo {author} {\bibfnamefont {O.}~\bibnamefont
  {G{\"u}hne}}\ and\ \bibinfo {author} {\bibfnamefont {G.}~\bibnamefont
  {T{\'o}th}},\ }\href {\doibase 10.1016/j.physrep.2009.02.004} {\bibfield
  {journal} {\bibinfo  {journal} {Phys. Rep.}\ }\textbf {\bibinfo {volume}
  {474}},\ \bibinfo {pages} {1} (\bibinfo {year} {2009})}\BibitemShut {NoStop}%
\bibitem [{\citenamefont {Bartkiewicz}\ \emph {et~al.}(2015)\citenamefont
  {Bartkiewicz}, \citenamefont {Beran}, \citenamefont {Lemr}, \citenamefont
  {Norek},\ and\ \citenamefont {Miranowicz}}]{Bartkiewicz}%
  \BibitemOpen
  \bibfield  {author} {\bibinfo {author} {\bibfnamefont {K.}~\bibnamefont
  {Bartkiewicz}}, \bibinfo {author} {\bibfnamefont {J.}~\bibnamefont {Beran}},
  \bibinfo {author} {\bibfnamefont {K.}~\bibnamefont {Lemr}}, \bibinfo {author}
  {\bibfnamefont {M.}~\bibnamefont {Norek}}, \ and\ \bibinfo {author}
  {\bibfnamefont {A.}~\bibnamefont {Miranowicz}},\ }\href {\doibase
  10.1103/physreva.91.022323} {\bibfield  {journal} {\bibinfo  {journal} {Phys.
  Rev. A}\ }\textbf {\bibinfo {volume} {91}},\ \bibinfo {pages} {022323}
  (\bibinfo {year} {2015})}\BibitemShut {NoStop}%
\bibitem [{\citenamefont {Knips}\ \emph {et~al.}(2020)\citenamefont {Knips},
  \citenamefont {Dziewior}, \citenamefont {K{\l}obus}, \citenamefont
  {Laskowski}, \citenamefont {Paterek}, \citenamefont {Shadbolt}, \citenamefont
  {Weinfurter},\ and\ \citenamefont {Meinecke}}]{knips2020multipartite}%
  \BibitemOpen
  \bibfield  {author} {\bibinfo {author} {\bibfnamefont {L.}~\bibnamefont
  {Knips}}, \bibinfo {author} {\bibfnamefont {J.}~\bibnamefont {Dziewior}},
  \bibinfo {author} {\bibfnamefont {W.}~\bibnamefont {K{\l}obus}}, \bibinfo
  {author} {\bibfnamefont {W.}~\bibnamefont {Laskowski}}, \bibinfo {author}
  {\bibfnamefont {T.}~\bibnamefont {Paterek}}, \bibinfo {author} {\bibfnamefont
  {P.~J.}\ \bibnamefont {Shadbolt}}, \bibinfo {author} {\bibfnamefont
  {H.}~\bibnamefont {Weinfurter}}, \ and\ \bibinfo {author} {\bibfnamefont
  {J.~D.}\ \bibnamefont {Meinecke}},\ }\href {\doibase
  10.1038/s41534-020-0281-5} {\bibfield  {journal} {\bibinfo  {journal} {npj
  Quantum Inf.}\ }\textbf {\bibinfo {volume} {6}},\ \bibinfo {pages} {51}
  (\bibinfo {year} {2020})}\BibitemShut {NoStop}%
\bibitem [{\citenamefont {Knips}(2020)}]{knips2020moment}%
  \BibitemOpen
  \bibfield  {author} {\bibinfo {author} {\bibfnamefont {L.}~\bibnamefont
  {Knips}},\ }\href {\doibase 10.22331/qv-2020-11-19-47} {\bibfield  {journal}
  {\bibinfo  {journal} {Quantum Views}\ }\textbf {\bibinfo {volume} {4}},\
  \bibinfo {pages} {47} (\bibinfo {year} {2020})}\BibitemShut {NoStop}%
\bibitem [{\citenamefont {Imai}\ \emph {et~al.}(2021)\citenamefont {Imai},
  \citenamefont {Wyderka}, \citenamefont {Ketterer},\ and\ \citenamefont
  {G{\"u}hne}}]{imai2021bound}%
  \BibitemOpen
  \bibfield  {author} {\bibinfo {author} {\bibfnamefont {S.}~\bibnamefont
  {Imai}}, \bibinfo {author} {\bibfnamefont {N.}~\bibnamefont {Wyderka}},
  \bibinfo {author} {\bibfnamefont {A.}~\bibnamefont {Ketterer}}, \ and\
  \bibinfo {author} {\bibfnamefont {O.}~\bibnamefont {G{\"u}hne}},\ }\href
  {\doibase 10.1103/PhysRevLett.126.150501} {\bibfield  {journal} {\bibinfo
  {journal} {Phys. Rev. Lett.}\ }\textbf {\bibinfo {volume} {126}},\ \bibinfo
  {pages} {150501} (\bibinfo {year} {2021})}\BibitemShut {NoStop}%
\bibitem [{\citenamefont {Elben}\ \emph {et~al.}(2023)\citenamefont {Elben},
  \citenamefont {Flammia}, \citenamefont {Huang}, \citenamefont {Kueng},
  \citenamefont {Preskill}, \citenamefont {Vermersch},\ and\ \citenamefont
  {Zoller}}]{elben2023randomized}%
  \BibitemOpen
  \bibfield  {author} {\bibinfo {author} {\bibfnamefont {A.}~\bibnamefont
  {Elben}}, \bibinfo {author} {\bibfnamefont {S.~T.}\ \bibnamefont {Flammia}},
  \bibinfo {author} {\bibfnamefont {H.-Y.}\ \bibnamefont {Huang}}, \bibinfo
  {author} {\bibfnamefont {R.}~\bibnamefont {Kueng}}, \bibinfo {author}
  {\bibfnamefont {J.}~\bibnamefont {Preskill}}, \bibinfo {author}
  {\bibfnamefont {B.}~\bibnamefont {Vermersch}}, \ and\ \bibinfo {author}
  {\bibfnamefont {P.}~\bibnamefont {Zoller}},\ }\href {\doibase
  10.1038/s42254-022-00535-2} {\bibfield  {journal} {\bibinfo  {journal} {Nat.
  Rev. Phys.}\ }\textbf {\bibinfo {volume} {5}},\ \bibinfo {pages} {9}
  (\bibinfo {year} {2023})}\BibitemShut {NoStop}%
\bibitem [{\citenamefont {Kempe}(1999)}]{kempe1999multiparticle}%
  \BibitemOpen
  \bibfield  {author} {\bibinfo {author} {\bibfnamefont {J.}~\bibnamefont
  {Kempe}},\ }\href {\doibase 10.1103/PhysRevA.60.910} {\bibfield  {journal}
  {\bibinfo  {journal} {Phys. Rev. A}\ }\textbf {\bibinfo {volume} {60}},\
  \bibinfo {pages} {910} (\bibinfo {year} {1999})}\BibitemShut {NoStop}%
\bibitem [{\citenamefont {Barnum}\ and\ \citenamefont {Linden}(2001)}]{Barnum}%
  \BibitemOpen
  \bibfield  {author} {\bibinfo {author} {\bibfnamefont {H.}~\bibnamefont
  {Barnum}}\ and\ \bibinfo {author} {\bibfnamefont {N.}~\bibnamefont
  {Linden}},\ }\href {\doibase 10.1088/0305-4470/34/35/305} {\bibfield
  {journal} {\bibinfo  {journal} {J. Phys. A: Math. Gen.}\ }\textbf {\bibinfo
  {volume} {34}},\ \bibinfo {pages} {6787} (\bibinfo {year}
  {2001})}\BibitemShut {NoStop}%
\bibitem [{\citenamefont {Wyderka}\ and\ \citenamefont
  {Ketterer}(2023)}]{Wyderka_Geometry}%
  \BibitemOpen
  \bibfield  {author} {\bibinfo {author} {\bibfnamefont {N.}~\bibnamefont
  {Wyderka}}\ and\ \bibinfo {author} {\bibfnamefont {A.}~\bibnamefont
  {Ketterer}},\ }\href {\doibase 10.1103/prxquantum.4.020325} {\bibfield
  {journal} {\bibinfo  {journal} {PRX Quantum}\ }\textbf {\bibinfo {volume}
  {4}},\ \bibinfo {pages} {020325} (\bibinfo {year} {2023})}\BibitemShut
  {NoStop}%
\bibitem [{\citenamefont {Wyderka}\ \emph
  {et~al.}(2023{\natexlab{b}})\citenamefont {Wyderka}, \citenamefont
  {Ketterer}, \citenamefont {Imai}, \citenamefont {B{\"o}nsel}, \citenamefont
  {Jones}, \citenamefont {Kirby}, \citenamefont {Yu},\ and\ \citenamefont
  {G{\"u}hne}}]{wyderka2023complete}%
  \BibitemOpen
  \bibfield  {author} {\bibinfo {author} {\bibfnamefont {N.}~\bibnamefont
  {Wyderka}}, \bibinfo {author} {\bibfnamefont {A.}~\bibnamefont {Ketterer}},
  \bibinfo {author} {\bibfnamefont {S.}~\bibnamefont {Imai}}, \bibinfo {author}
  {\bibfnamefont {J.~L.}\ \bibnamefont {B{\"o}nsel}}, \bibinfo {author}
  {\bibfnamefont {D.~E.}\ \bibnamefont {Jones}}, \bibinfo {author}
  {\bibfnamefont {B.~T.}\ \bibnamefont {Kirby}}, \bibinfo {author}
  {\bibfnamefont {X.-D.}\ \bibnamefont {Yu}}, \ and\ \bibinfo {author}
  {\bibfnamefont {O.}~\bibnamefont {G{\"u}hne}},\ }\href {\doibase
  10.1103/physrevlett.131.090201} {\bibfield  {journal} {\bibinfo  {journal}
  {Phys. Rev. Lett.}\ }\textbf {\bibinfo {volume} {131}},\ \bibinfo {pages}
  {090201} (\bibinfo {year} {2023}{\natexlab{b}})}\BibitemShut {NoStop}%
\bibitem [{\citenamefont {Ma}\ and\ \citenamefont
  {Huang}(2025)}]{ma2025construct}%
  \BibitemOpen
  \bibfield  {author} {\bibinfo {author} {\bibfnamefont {F.}~\bibnamefont
  {Ma}}\ and\ \bibinfo {author} {\bibfnamefont {H.-Y.}\ \bibnamefont {Huang}},\
  }in\ \href {\doibase 10.1145/3717823.3718254} {\emph {\bibinfo {booktitle}
  {Proceedings of the 57th Annual ACM Symposium on Theory of Computing}}}\
  (\bibinfo {year} {2025})\ pp.\ \bibinfo {pages} {806--809}\BibitemShut
  {NoStop}%
\bibitem [{\citenamefont {Schmidt}(1907)}]{schmidt1907theorie}%
  \BibitemOpen
  \bibfield  {author} {\bibinfo {author} {\bibfnamefont {E.}~\bibnamefont
  {Schmidt}},\ }\href@noop {} {\bibfield  {journal} {\bibinfo  {journal}
  {Mathematische Annalen}\ }\textbf {\bibinfo {volume} {63}},\ \bibinfo {pages}
  {433} (\bibinfo {year} {1907})}\BibitemShut {NoStop}%
\bibitem [{\citenamefont {Grassl}\ \emph {et~al.}(1998)\citenamefont {Grassl},
  \citenamefont {R{\"o}tteler},\ and\ \citenamefont
  {Beth}}]{grassl1998computing}%
  \BibitemOpen
  \bibfield  {author} {\bibinfo {author} {\bibfnamefont {M.}~\bibnamefont
  {Grassl}}, \bibinfo {author} {\bibfnamefont {M.}~\bibnamefont
  {R{\"o}tteler}}, \ and\ \bibinfo {author} {\bibfnamefont {T.}~\bibnamefont
  {Beth}},\ }\href {\doibase 10.1103/PhysRevA.58.1833} {\bibfield  {journal}
  {\bibinfo  {journal} {Phys. Rev. A}\ }\textbf {\bibinfo {volume} {58}},\
  \bibinfo {pages} {1833} (\bibinfo {year} {1998})}\BibitemShut {NoStop}%
\bibitem [{\citenamefont {Onishchik}\ and\ \citenamefont
  {Vinberg}(2012)}]{onishchik2012lie}%
  \BibitemOpen
  \bibfield  {author} {\bibinfo {author} {\bibfnamefont {A.~L.}\ \bibnamefont
  {Onishchik}}\ and\ \bibinfo {author} {\bibfnamefont {E.~B.}\ \bibnamefont
  {Vinberg}},\ }\href {\doibase 10.1007/978-3-642-74334-4} {\emph {\bibinfo
  {title} {Lie groups and algebraic groups}}}\ (\bibinfo  {publisher} {Springer
  Science \& Business Media},\ \bibinfo {year} {2012})\BibitemShut {NoStop}%
\bibitem [{\citenamefont {Springer}(2006)}]{springer2006invariant}%
  \BibitemOpen
  \bibfield  {author} {\bibinfo {author} {\bibfnamefont {T.~A.}\ \bibnamefont
  {Springer}},\ }\href {\doibase 10.1007/BFb0095644} {\emph {\bibinfo {title}
  {Invariant theory}}},\ Vol.\ \bibinfo {volume} {585}\ (\bibinfo  {publisher}
  {Springer},\ \bibinfo {year} {2006})\BibitemShut {NoStop}%
\bibitem [{\citenamefont {Makhlin}(2002)}]{Makhlin}%
  \BibitemOpen
  \bibfield  {author} {\bibinfo {author} {\bibfnamefont {Y.}~\bibnamefont
  {Makhlin}},\ }\href {\doibase 10.1023/a:1022144002391} {\bibfield  {journal}
  {\bibinfo  {journal} {Quantum Inf. Process.}\ }\textbf {\bibinfo {volume}
  {1}},\ \bibinfo {pages} {243–252} (\bibinfo {year} {2002})}\BibitemShut
  {NoStop}%
\bibitem [{\citenamefont {Vidal}\ and\ \citenamefont
  {Werner}(2002)}]{vidal2002computable}%
  \BibitemOpen
  \bibfield  {author} {\bibinfo {author} {\bibfnamefont {G.}~\bibnamefont
  {Vidal}}\ and\ \bibinfo {author} {\bibfnamefont {R.~F.}\ \bibnamefont
  {Werner}},\ }\href {\doibase 10.1103/PhysRevA.65.032314} {\bibfield
  {journal} {\bibinfo  {journal} {Phys. Rev. A}\ }\textbf {\bibinfo {volume}
  {65}},\ \bibinfo {pages} {032314} (\bibinfo {year} {2002})}\BibitemShut
  {NoStop}%
\bibitem [{\citenamefont {Werner}(1989)}]{werner1989quantum}%
  \BibitemOpen
  \bibfield  {author} {\bibinfo {author} {\bibfnamefont {R.~F.}\ \bibnamefont
  {Werner}},\ }\href {\doibase 10.1103/PhysRevA.40.4277} {\bibfield  {journal}
  {\bibinfo  {journal} {Phys. Rev. A}\ }\textbf {\bibinfo {volume} {40}},\
  \bibinfo {pages} {4277} (\bibinfo {year} {1989})}\BibitemShut {NoStop}%
\bibitem [{\citenamefont {Gross}\ \emph {et~al.}(2007)\citenamefont {Gross},
  \citenamefont {Audenaert},\ and\ \citenamefont {Eisert}}]{gross2007evenly}%
  \BibitemOpen
  \bibfield  {author} {\bibinfo {author} {\bibfnamefont {D.}~\bibnamefont
  {Gross}}, \bibinfo {author} {\bibfnamefont {K.}~\bibnamefont {Audenaert}}, \
  and\ \bibinfo {author} {\bibfnamefont {J.}~\bibnamefont {Eisert}},\ }\href
  {\doibase 10.1063/1.2716992} {\bibfield  {journal} {\bibinfo  {journal} {J.
  Math. Phys.}\ }\textbf {\bibinfo {volume} {48}},\ \bibinfo {pages} {052104}
  (\bibinfo {year} {2007})}\BibitemShut {NoStop}%
\bibitem [{\citenamefont {Köstenberger}(2021)}]{Wein}%
  \BibitemOpen
  \bibfield  {author} {\bibinfo {author} {\bibfnamefont {G.}~\bibnamefont
  {Köstenberger}},\ }\href {https://doi.org/10.48550/arXiv.2101.00921}
  {\bibfield  {journal} {\bibinfo  {journal} {arXiv preprint arXiv:2101.00921}\
  } (\bibinfo {year} {2021})}\BibitemShut {NoStop}%
\bibitem [{\citenamefont {Aubrun}(2018)}]{SWD}%
  \BibitemOpen
  \bibfield  {author} {\bibinfo {author} {\bibfnamefont {G.}~\bibnamefont
  {Aubrun}},\ }\href@noop {} {\enquote {\bibinfo {title} {A naive look at
  {S}chur–{W}eyl duality},}\ } (\bibinfo {year} {2018}),\ \bibinfo {note}
  {\url{https://math.univ-lyon1.fr/~aubrun/recherche/schur-weyl.pdf} [Accessed:
  09.04.2026]}\BibitemShut {NoStop}%
\end{thebibliography}%

\appendix

\onecolumngrid
\parindent 0pt
\section{Notation and general calculation of integrals with respect to the Haar measure}\label{Haar Calculation}
\begin{definition}
    The twirling operator (see \cite{werner1989quantum,gross2007evenly} ) is given by:   
    \begin{equation}
    \begin{split}
        \Phi^{(t)} : L(\mathbb{C}^d)^{\otimes n} &\to \left(L(\mathbb{C}^d)^{\otimes n}\right)^{\otimes t} \\
        \mathcal{O} &\mapsto \int (U_1^\dagger \otimes ... \otimes U_n^\dagger)^{\otimes t} \mathcal{O}^{\otimes t} (U_1 \otimes ... \otimes U_n)^{\otimes t} \dd (U_1,...,U_n). 
    \end{split}
    \end{equation}
    Here, $L(V)$ denotes the (continuous) space of linear maps over a vector space $V$ and the integral is with respect to the product of Haar measures over $SU(d)$. Furthermore, following Ref.~\cite{Wein}, let 
    \begin{equation}
    \begin{split}
        P : L(\mathbb{C}^d)^{\otimes t} &\to L(\mathbb{C}^d)^{\otimes t} \\
        A &\mapsto \int (U^{\dagger})^{\otimes t}A U^{\otimes t} \dd U.
    \end{split}
    \end{equation}
    
\end{definition}
$\Phi^{(t)}$ is closely linked to the moments $\mathcal{R}^{(t)}_{\mathcal{O}} ( \varrho)$ from Equation (\ref{moments}) via $\mathcal{R}^{(t)}_{\mathcal{O}} ( \varrho) = \tr( \varrho^{\otimes t} \Phi^{(t)}(\mathcal{O}))$, which follows from the fact that $\tr(M)^t = \tr(M^{\otimes t})$ for any matrix $M$.

Consider an arbitrary operator $\mathcal{O} = \sum\limits_{j=1}^r \bigotimes\limits_{k=1}^n A_j^{(k)} \in L(\mathbb{C}^d)^{\otimes n}$. Then

\begin{equation}\label{reduction}
\begin{split}
    \Phi^{(t)}(\mathcal{O}) &= \int (U_1^\dagger \otimes ... \otimes U_n^\dagger)^{\otimes t} \left( \sum_{j=1}^r A_j^{(1)} \otimes ... \otimes A_j^{(n)} \right)^{\otimes t} (U_1 \otimes ...\otimes U_n) \; \dd (U_1,...,U_n) \\
    &=\sum_{j_1,...,j_t=1}^r \bigotimes_{k=1}^n \int (U_k^\dagger)^{\otimes t} \left(\bigotimes_{\ell = 1}^t A_{j_\ell}^{(k)} \right)  U_k^{\otimes t} \dd U_k = \sum_{j_1,...,j_t=1}^r \bigotimes_{k=1}^n P\left(\bigotimes_{\ell = 1}^t A_{j_\ell}^{(k)} \right).
\end{split}
\end{equation}

Thus the twirling operator can be evaluated with the help of the $P$ function.

By construction, $P$ is linear. Furthermore, for all $A \in L(\mathbb{C}^d)^{\otimes t}$, the operator $P(A)$ commutes with all elements of the form $U^{\otimes t}$.

\begin{definition}\label{V_pi def}
    For $t\in \mathbb{N}$ denote the group of permutations of $\{1,...,t\}$ as $S_t$. For $\pi \in S_t$ let $V_\pi \in L(\mathbb{C}^d)^{\otimes t}$ be defined via 
\begin{equation}\label{V_pi}
    V_\pi\ket{e_{j_1}} \ket{e_{j_2}} ...\ket{e_{j_t}} = \ket{e_{j_{\pi(1)}}}...\ket{e_{j_{\pi(t)}}} \quad \forall \vec{j} \in \{1,...,d\}^t.
\end{equation}
\end{definition}

\begin{theorem}[Schur-Weyl duality]\label{SW}
    Let $t,d \in \mathbb{N}$. The set of elements that commute with $\mathrm{span}(\{U^{\otimes t} : U \in U(d)\})$ is given by $\mathrm{span}(\{V_\pi : \pi \in S_t\})$.
    This implies that for any $A \in L(\mathbb{C}^d)^{\otimes t}$, there exist of coefficients $c_\pi \in \mathbb{C}$, such that
    \begin{equation}
        P(A) = \sum\limits_{\pi \in S_t} c_\pi V_\pi.
    \end{equation}
\end{theorem}
\begin{proof}
    This formulation of the theorem is taken from \cite{SWD}, where a proof of it can be found.
\end{proof}

Consider a bipartite Hermitian operator in its operator Schmidt decomposition 

$\mathcal{O} = \sum\limits_{j=1}^r s_j A_j \otimes B_j \in L(\mathbb{C}^d)^{\otimes 2}_{\mathrm{herm.}} = \{ M \in \mathbb{C}^{d\times d} \; | \; M^\dagger =M \}^{\otimes 2}$,
where $A_j$ and $B_j$ are Hermitian, satisfying 

$\tr(A_jA_k)=\delta_{jk} = \tr(B_jB_k)$ and $s_j>0$. Then Equation (\ref{reduction}) reduces to:
\begin{equation}\label{MultInt1}
\begin{split}
    \Phi^{(t)}(\mathcal{O})
    =& \sum_{j_1,...,j_t = 1}^r \left(\prod_{k=1}^t s_{j_k}\right) P\left(\bigotimes_{k=1}^t A_{j_k} \right) \otimes P\left(\bigotimes_{k=1}^t B_{j_k} \right) \\
    =& \sum_{j_1,...,j_t = 1}^r \left(\prod_{k=1}^t s_{j_k}\right) \left(\sum_{\pi_A\in S_t} x^{(\,\vec{j} \, )}_{\pi_A} \; V_{\pi_A}\right) \otimes \left(\sum_{\pi_B \in S_t} y^{(\,\vec{j}\,)}_{\pi_B} \; V_{\pi_B}\right) \\
    =& \sum_{\pi_A,\pi_B \in S_t}\sum_{j_1,...,j_t = 1}^r \left(\prod_{k=1}^t s_{j_k}\right) x^{(\,\vec{j} \, )}_{\pi_A} y^{(\,\vec{j}\,)}_{\pi_B} \; V_{\pi_A, \pi_B} \equiv \sum_{\pi_A,\pi_B \in S_t} c_{\pi_A,\pi_B} V_{\pi_A,\pi_B}.
\end{split}
\end{equation}
where $x^{\,(\vec{j} \, )}_{\pi_A} \in \mathbb{C}$ and $y^{\,(\vec{j} \, )}_{\pi_B} \in \mathbb{C}$ are the coefficients from the Schur-Weyl duality and $V_{\pi_A, \pi_B} :=V_{\pi_A} \otimes V_{\pi_B} $.

This reduces the problem of calculating $\mathcal{R}^{(t)}_{\mathcal{O}}$ to finding the prefactors $c_{\pi_A,\pi_B}$ and calculating $\tr(\varrho^{\otimes t} V_{\pi_A,\pi_B})$ for all $\pi_A,\pi_B \in S_t$. Both require the following Lemma.

\begin{lemma}\label{calc}
    For $t,d\in \mathbb{N}$, $A_1,...,A_t \in \mathbb{C}^{d\times d}$, the following identity holds:
    \begin{equation}
        \tr(A_1\otimes...\otimes A_t V_\pi) = \prod_{\substack{j=1, \\ j \notin \{\pi^k(\ell) : \ell < j, k \leq t\}}}^t\tr(A_{\pi(j)}A_{\pi^2(j)}...A_j).
    \end{equation}
    Written in another way: If $\pi$ consists of the $k$ cycles starting at $m_1,...,m_k$, then
    \begin{equation}
        \tr(A_1\otimes...\otimes A_t V_\pi) = \prod_{\ell = 1}^k \tr(A_{\pi(m_\ell)}A_{\pi^2(m_\ell)} ... A_{m_\ell}).
    \end{equation}
    In the special case of $A\equiv A_1 = ... = A_t$
    \begin{equation}
        \tr(A^{\otimes t} V_\pi) = \prod_{j=1}^t \tr(A^j)^{\mathrm{\#cycles \; of \; length\;} j}.
    \end{equation}
    In the special case of $A \equiv \mathbbm{1}$
    \begin{equation}
        \tr(V_\pi) = d^{\mathrm{\#\, cycles}}.
    \end{equation}
\end{lemma}
\begin{proof}
    Let $\pi$ consist of $k$ cycles of cycle lengths $(m_j)_j$ starting at $(\ell_j)_j$. For ease of notation denote $A_j$ with $A^{(j)}$ instead. Then
    \begin{equation*}
    \begin{split}
        &\tr(A^{(1)}\otimes...\otimes A^{(t)} V_\pi) = \sum_{i_1,...,i_t = 1}^d A^{(1)}_{i_1,i_{\pi(1)}}\cdot ... \cdot A^{(t)}_{i_ti_{\pi(t)}} \\
        &= \sum_{i_1,...,i_t = 1}^d A^{(1)}_{i_1,i_{\pi(1)}} A^{(\pi(1))}_{i_{\pi(1)}, i_{\pi^2(1)}}  A^{(\pi^2(1))}_{i_{\pi^2(1), \pi^3(1)}} \cdot ... \cdot A^{(\pi^{\ell_1-1}(1))}_{i_{\pi^{\ell_1-1}(1)},i_1} A^{(m_2)}_{i_{m_2},i_{\pi{m_2}}} \cdot ... \cdot A^{(\pi^{\ell_2-1}(m_2))}_{i_{\pi^{\ell_2 - 1}(m_2)},i_{m_2}} \cdot ... \cdot A^{(\pi^{\ell_k - 1}(m_k))}_{i_{\pi^{\ell_k - 1}(m_k)},i_{m_k}} \\
        &= \sum_{i_1,...,i_t = 1}^d \prod_{n=1}^k\prod_{j = 0}^{\ell_n-1} A^{(\pi^{j}(m_n))}_{i_{\pi^{j}(m_n)},i_{\pi^{j+1}(m_n)}} = \prod_{n=1}^k \left(\sum_{\tilde{i}_1,...,\tilde{i}_{\ell_n} = 1}^d \prod_{j = 0}^{\ell_n-1} A^{(\pi^{j}(m_n))}_{\tilde{i}_j,\tilde{i}_{j+1}}\right) \\
        &= \prod_{n = 1}^k \tr(A^{(m_n)}A^{(\pi(m_n))}A^{(\pi^2(m_n))} \cdot ... \cdot A^{(\pi^{\ell_n - 1}(m_n))}),
    \end{split}
    \end{equation*}

    The other statements are just reformulations of the one proven here. The last one uses the fact that $\tr(\mathbbm{1}^j) = \tr(\mathbbm{1}) = d$.

\end{proof}
Note that the Schur-Weyl duality implies $\tr(\mathbb{E}(A)V_{\pi}) = \sum\limits_{\pi' \in S_t} c_{\pi'} \tr(V_{\pi'\circ \pi})$ while the definition of $P$ implies 

$\tr(P(A)V_\pi) = \int_{U(d)} \tr( (U^\dagger)^{\otimes t} A U^{\otimes t} V_\pi) \dd U = \tr(A V_\pi)$, since $V_\pi$ commutes with $U^{\otimes t}$. 

Therefore, $ \sum\limits_{\pi' \in S_t} c_{\pi'} \tr(V_{\pi'\circ \pi}) \stackrel{!}{=}\tr(A V_\pi)$, where both traces can be computed using Lemma \ref{calc}. This leads to a linear system of equations, that $(c_\pi)_{\pi\in S_t}$ needs to satisfy.

\begin{proposition}\label{explicit}
    Let $t,d\in \mathbb{N}$. Let $\mathcal{O}\in L(\mathbb{C}^d)^{\otimes 2}_{\mathrm{herm.}}$ be written in its operator Schmidt decomposition $\mathcal{O} = \sum\limits_{j=1}^r s_j A_j \otimes B_j$ and let $M^{(t)}_d = (\tr(V_{\pi \circ \pi'}))_{\pi, \pi' \in S_t}$. Then any bipartite quantum state $\varrho \in L(\mathbb{C}^d)^{\otimes 2}_{\mathrm{herm.}}$ fulfills
    \begin{equation}\label{twirl eq}
        \mathcal{R}^{(t)}_{\mathcal{O}}(\varrho)=\sum_{\pi_A,\pi_B \in S_t} c_{\pi_A,\pi_B} \tr(\varrho^{\otimes t} V_{\pi_A} \otimes V_{\pi_B}).
    \end{equation}
    where $c_{\pi_A,\pi_B} = \sum\limits_{j_1,...,j_t = 1}^r \left(\prod\limits_{k=1}^t s_{j_k}\right) x^{(\,\vec{j} \, )}_{\pi_A} y^{(\,\vec{j}\,)}_{\pi_B}$ for some $x^{(\,\vec{j} \, )}, y^{(\,\vec{j}\,)} \in \mathbb{C}^{S_t}$ that satisfy 
    \begin{equation}\label{lin eq}
        M^{(t)}_d x^{(\,\vec{j} \, )} = (\tr(A_{j_1} \otimes ... \otimes A_{j_t} V_\pi))_{\pi\in S_t}, \quad M^{(t)}_d y^{(\,\vec{j} \, )} = (\tr(B_{j_1} \otimes ... \otimes B_{j_t} V_\pi))_\pi.
    \end{equation}
    If the operator Schmidt decomposition of $\mathcal{O}$ happens to be symmetric (i.e. $A_j = B_j$ for all $j=1,...,r$), then $c_{\pi_A,\pi_B} = c_{\pi_B,\pi_A}$.
\end{proposition}
Throughout the rest of the paper, $S_3$ shall be ordered as $( (),(12),(13),(23),(123),(132))$.
\begin{examp}\label{t=3 examp}
Using this ordering of $S_3$ one obtains
\begin{equation} 
    M^{(3)}_d\ =
    \left(\begin{array} {cccccc}
    d^3 & d^2 & d^2 & d^2 & d & d \\
    d^2 & d^3 & d & d & d^2 & d^2 \\
    d^2 & d & d^3 & d & d^2 & d^2 \\
    d^2 & d & d & d^3 & d^2 & d^2 \\
    d & d^2 & d^2 & d^2 & d & d^3 \\
    d & d^2 & d^2 & d^2 & d^3 & d \\
    \end{array} \right)
    \end{equation}
and 
\begin{equation}
    \tr(A_{j_1} \otimes ... \otimes A_{j_3} V_\pi)_{\pi\in S_t} = \begin{bmatrix}
    \,\tr(A_{j_1})\tr(A_{j_2})\tr(A_{j_3})\\
    \tr(A_{j_3})\delta_{j_1j_2}\\
    \tr(A_{j_2})\delta_{j_1j_3}\\
    \tr(A_{j_1})\delta_{j_2j_3}\\
    \tr(A_{j_1}A_{j_2}A_{j_3})\\
    \tr(A_{j_1}A_{j_3}A_{j_2})\\
    \end{bmatrix}.
\end{equation}
\end{examp}
Schur-Weyl duality only guarantees the existence of prefactors $(c_\pi)_{\pi\in S_t}$ that need to satisfy (\ref{lin eq}). If there is no unique solution to (\ref{lin eq}), more care has to be taken to determine which of the solutions corresponds to the actual prefactors. Furthermore, in low dimensions the prefactors often cannot be uniquely determined, since the set $\{V_\pi : \pi \in S_t\}$ can be linearly dependent (a trivial example is the $d=1, t=2$ case, where $V_{()} = V_{(12)}$). This means that the dimension of the solution space (i.e., the kernel of $M_d^{(t)}$) is at least as large as the number of independent symmetries. Luckily, in the situations that will be considered in this paper, the solution space can be shown to be precisely generated by the symmetries of $(V_\pi)_{\pi \in S_t}$, which implies that every solution of the linear equation corresponds to a set of correct prefactors $(c_\pi)_{\pi \in S_t}$. This is done by the following lemma, which is proven by explicit calculation.

\begin{lemma}\label{sym}
    Consider the set $\{V_\pi \in L(\mathbb{C}^d)^{\otimes t} : \pi \in S_t\}$ from Definition \ref{V_pi def}.
    For $d=2, t=3$, the following identity holds:
        \begin{equation}
        V_{(12)} + V_{(13)} + V_{(23)} - V_{(123)} - V_{(132)} = \mathbbm{1} = V_{()}.
    \end{equation}
    The kernel of $M^{(3)}_2$ (defined in Proposition \ref{explicit}) is one dimensional and spanned by $(1,-1,-1,-1,1,1)^{\T} $. This implies that any solution to (\ref{lin eq}) yields a correct solution for (\ref{twirl eq}). There always exists a unique set of prefactors  $(c_\pi)_{\pi \in S_t}$ satisfying (\ref{lin eq}) such that $c_{(132)} = 0$.

    Now let $d=2$ and $t=4$. $M_2^{(4)}$ is a $24 \times 24$ matrix with a $10$ dimensional kernel, which corresponds directly to the ten symmetries
        \begin{equation*}
        \begin{split}
        3V_{()}&-V_{(12)}-2V_{(13)}-V_{(14)}-V_{(23)}-2V_{(24)}-V_{(34)}-V_{(12)(34)}+V_{(13)(24)}-V_{(14)(23)}\\&+V_{(123)}+V_{(124)}+V_{(134)}+V_{(234)}+2V_{(1432)}=0 \\
        V_{()}&-V_{(12)}-V_{(14)}+V_{(23)}-V_{(34)}+V_{(12)(34)}-V_{(13)(24)}-V_{(14)(23)}-V_{(123)}+V_{(124)}+V_{(134)}-V_{(234)}+2V_{(1423)}=0 \\
        V_{()}&-V_{(12)}-V_{(14)}-V_{(23)}+V_{(34)}-V_{(12)(34)}-V_{(13)(24)}+V_{(14)(23)}+V_{(123)}+V_{(124)}-V_{(134)}-V_{(234)}+2V_{(1342)}=0 \\
        V_{()}&-V_{(12)}+V_{(14)}-V_{(23)}-V_{(34)}+V_{(12)(34)}-V_{(13)(24)}-V_{(14)(23)}+V_{(123)}-V_{(124)}-V_{(134)}+V_{(234)}+2V_{(1324)}=0 \\
        V_{()}&+V_{(12)}-V_{(14)}-V_{(23)}-V_{(34)}-V_{(12)(34)}-V_{(13)(24)}+V_{(14)(23)}-V_{(123)}-V_{(124)}+V_{(134)}+V_{(234)}+2V_{(1243)}=0 \\
        -V_{()}&+V_{(12)}+V_{(14)}+V_{(23)}+V_{(34)}-V_{(12)(34)}+V_{(13)(24)}-V_{(14)(23)}-V_{(123)}-V_{(124)}-V_{(134)}-V_{(234)}+2V_{(1234)}=0 \\
        V_{()}&-V_{(23)}-V_{(24)}-V_{(34)}+V_{(234)}+V_{(243)}=0 \\
        V_{()}&-V_{(13)}-V_{(14)}-V_{(34)}+V_{(134)}+V_{(143)}=0 \\
        V_{()}&-V_{(12)}-V_{(14)}-V_{(24)}+V_{(124)}+V_{(142)}=0 \\
        V_{()}&-V_{(12)}-V_{(13)}-V_{(23)}+V_{(123)}+V_{(132)}=0.
        \end{split}
    \end{equation*}
    These can be used to always set $c_{(132)},c_{(124)},c_{(142)} c_{(134)}, c_{(143)}, c_{(234)},c_{(243)},c_{(12)(34)},c_{(13)(24)}$ and $c_{(14)(23)}$ to zero which uniquely determines the other prefactors. Additionally, any solution to (\ref{lin eq}) yields a correct solution for (\ref{twirl eq}).

\end{lemma}

\begin{examp}
    Choosing $x^{( \, \vec{j} \, )}_{(132)} \equiv 0 $ in the two dimensional bipartite case modifies the linear equation from Example \ref{t=3 examp} to be
    \begin{equation}\label{t=3 mat}
    \left(\begin{array} {cccccc}
    8 & 4 & 4 & 4 & 2  \\[6pt]
    4 & 8 & 2 & 2 & 4  \\[6pt]
    4 & 2 & 8 & 2 & 4  \\[6pt]
    4 & 2 & 2 & 8 & 4  \\[6pt]
    2 & 4 & 4 & 4 & 2  \\[6pt]
    \end{array} \right) \begin{bmatrix}
    x^{(\, \vec{j} \, )}_{()}\\
    x^{(\,\vec{j} \, )}_{(12)}\\
    x^{(\,\vec{j} \, )}_{(13)}\\
    x^{(\,\vec{j} \, )}_{(23)}\\
    x^{(\,\vec{j} \, )}_{(123)}\\
    \end{bmatrix} = \begin{bmatrix}
    \,\tr(A_{j_1})\tr(A_{j_2})\tr(A_{j_3})\\[6pt]
    \tr(A_{j_3})\delta_{j_1j_2}\\[6pt]
    \tr(A_{j_2})\delta_{j_1j_3}\\[6pt]
    \tr(A_{j_1})\delta_{j_2j_3}\\[6pt]
    \tr(A_{j_1}A_{j_2}A_{j_3})\\[6pt]
    \end{bmatrix}
    \end{equation}
    which can be inverted.
\end{examp}
Proposition \ref{explicit} additionally requires an explicit form for $\tr(\varrho^{\otimes 3}V_\pi)$. At least for the case of bipartite quantum states, they can be explicitly calculated using the Bloch representation
\begin{equation}\label{density matrix}
    \varrho = \frac{1}{4} \left[ \mathbbm{1} \otimes \mathbbm{1} + \sum_{k=1}^3 \alpha_k \sigma_k \otimes \mathbbm{1} + \sum_{k=1}^3 \beta_k \mathbbm{1} \otimes \sigma_k + \sum_{j,k=1}^3 T_{jk} \sigma_j \otimes \sigma_k \right],
\end{equation}
 (here $\vec{\alpha},\vec{\beta} \in \mathbb{R}^3, T\in \mathbb{R}^{3 \times 3}$ need to be chosen such that $\varrho \geq 0$) and applying Lemma \ref{calc}.

Direct calculation via Lemma \ref{calc} and the fact that for any matrices $A,B$ it holds that $\tr(A\otimes B) = \tr(A) \tr(B)$ yields:
\begin{lemma}\label{t3d2 system}
Let $\varrho$ be a bipartite quantum state in its Bloch representation (\ref{density matrix}) and consider $V_\pi \in L(\mathbb{C}^2)^{\otimes 3}$ from Definition \ref{V_pi def} (this corresponds to $d=2, t=3$). For $\pi,\pi' \in S_3$ let $V_{\pi,\pi'} = V_\pi \otimes V_{\pi'}$. Then
\begin{equation}\label{trace matrix}
    \begin{bmatrix}
        \tr(\varrho^{\otimes 3} V_{(),()}) \\
        \tr(\varrho^{\otimes 3} V_{(),(ij)}) \\
        \tr(\varrho^{\otimes 3} V_{(ij),()}) \\
        \tr(\varrho^{\otimes 3} V_{(),(123)}) \\
        \tr(\varrho^{\otimes 3} V_{(123),()}) \\
        \tr(\varrho^{\otimes 3} V_{(ij),(ij)}) \\
        \tr(\varrho^{\otimes 3} V_{(ij),(ik) \, j \neq k}) \\
        \tr(\varrho^{\otimes 3} V_{(ij),(123)}) \\
        \tr(\varrho^{\otimes 3} V_{(123),(ij)}) \\
        \tr(\varrho^{\otimes 3} V_{(123),(123)}) \\
        \end{bmatrix} = \frac{1}{16}
    \left(\begin{array} {ccccccc} 
    16 & 0 & 0 & 0 & 0 & 0 \\
    8 & 0 & 8 & 0 & 0 & 0 \\
    8 & 8 & 0 & 0 & 0 & 0  \\
    4 & 0 & 12 & 0 & 0 & 0 \\
    4 & 12 & 0 & 0 & 0 & 0 \\
    4 & 4 & 4 & 4 & 0 & 0 \\
    4 & 4 & 4 & 0 & 4 & 0 \\
    2 & 2 & 6 & 2 & 4 & 0 \\
    2 & 6 & 2 & 2 & 4 & 0 \\
    1 & 3 & 3 & 3 & 6 & -6 \\
    \end{array} \right)\begin{bmatrix}  
    1\\
    |\vec{\alpha}|^2\\
    |\vec{\beta}|^2\\
    \tr(T^{\T} T)\\
    \langle \vec{\alpha}, T \vec{\beta} \rangle \\
    \det(T)\\
    \end{bmatrix}.
\end{equation}

If on chooses an $\mathcal{O} \in L(\mathbb{C}^2)^{\otimes 2}_{\mathrm{herm.}}$ in Proposition \ref{explicit} with a symmetric operator Schmidt decomposition, then due to the symmetry of $c_{\alpha,\beta}$ and Equation (\ref{trace matrix}), $\mathcal{R}_{\mathcal{O}}^{(t)}(\varrho)$ can always be written as a linear combination of $f(\varrho) \in \mathbb{C}^7$ given by

\begin{equation}
    \underbrace{\left[\begin{array}{cl}
        &\tr(\varrho^{\otimes 3} V_{(),()}) \\
        \sum\limits_{\substack{i,j=1 \\ i<j}}^3 &\tr(\varrho^{\otimes 3} (V_{(),(ij)}+V_{(ij),()}))\\
        &\tr(\varrho^{\otimes 3} (V_{(),(123)}+V_{(123),()}))\\
        &\tr(\varrho^{\otimes 3} (V_{(11),(11)}+V_{(22),(22)}+V_{(33),(33)})) \\
        \sum\limits_{\substack{i < j, i < k, j \neq k}} &\tr(\varrho^{\otimes 3} V_{(ij),(ik)}) \\
        \sum\limits_{\substack{i,j=1 \\ j> i}}^3 &\tr(\varrho^{\otimes 3} (V_{(ij),(123)}+V_{(123),(ij)}))) \\
        &\tr(\varrho^{\otimes 3} V_{(123),(123)}) \end{array}\right]}_{\equiv f(\varrho)} = \underbrace{\frac{1}{16}
    \left(\begin{array} {ccccccc} 
    16 & 0 & 0 & 0 & 0 \\
    48 & 24 & 0 & 0 & 0 \\
    8 & 12 & 0 & 0 & 0 \\
    12 & 12 & 12 & 0 & 0 \\
    24 & 24 & 0 & 24 & 0 \\
    12 & 24 & 12 & 24 & 0 \\
    1 & 3 & 3 & 6 & -6 \\
    \end{array} \right)}_{\equiv D} \underbrace{\begin{bmatrix*}[c]
    1\\
    |\vec{\alpha}|^2+|\vec{\beta}|^2\\
    \tr(T^{\T} T)\\
    \langle \vec{\alpha}, T \vec{\beta} \rangle \\
    \det(T)\\
    \end{bmatrix*}}_{\equiv \tilde{I}(\varrho)}.
\end{equation}

\end{lemma}
\section{Obtaining a $\det(T)$ dependence}\label{app:detT}
Lemma \ref{t3d2 system} implies that any randomized measurement using (at most) third moments depends on the determinant if and only the prefactor in front of $\tr(\varrho^{\otimes 3} V_{(123),(123)})$ is nonzero. According to Proposition \ref{explicit}, $c_{(123),(123)} = \sum\limits_{j_1,j_2,j_3=1}^r s_{j_1} s_{j_2} s_{j_3} x_{(123)}^{( \, \vec{j} \, )} y_{(123)}^{( \, \vec{j} \, )}$ where $x_{(123)}^{( \, \vec{j} \,)}$ is obtained by inverting the matrix from Example \ref{t=3 examp}, leading to
\begin{equation}\label{TrA}
    x_{(123)}^{( \, \vec{j} \,)} =-2 \tr(A_{j_1})\tr(A_{j_2})\tr(A_{j_3}) + 2 \big[\tr(A_{j_1}) \delta_{j_2j_3} + \tr(A_{j_2}) \delta_{j_3j_1}+ \tr(A_{j_3}) \delta_{j_1j_2}\big] - 4 \tr(A_{j_1}A_{j_2}A_{j_3}).
\end{equation}
This is the key observation leading to:

\begin{theorem}\label{r>2}
    Consider a bipartite state $\varrho$ in its Bloch representation (\ref{density matrix}). There is no observable $\mathcal{O} \in  L(\mathbb{C}^2)^{\otimes 2}_{\mathrm{herm.}}$ of Schmidt rank $r \leq 2$, such that $\mathcal{R}^{(t)}_{\mathcal{O}}(\varrho)$ depends on $\det(T)$ for $t\leq 3$. Thus $I_1 = \det(T)$ is an invariant of type 3.
\end{theorem}
\begin{proof}
    For $t=1,2$, the determinant cannot appear since it is a polynomial in third powers of $T_{ij}$ (and thus in third powers of $\varrho_{ij}$). In general, polynomials in $k$-th powers of entries of $\varrho$ can only appear for $t\geq k$.
    \\
    
    For $t=3$, it only remains to prove that $c_{(123),(123)}$ is always $0$ for $r\in\{1,2\}$. For this, it is sufficient to show that $x_{(123)}^{( \, \vec{j} \, )}$ always vanishes.
    Since each entry of $\vec{j}$ is either $1$ or $2$, there are now two cases to consider:
    
    (1) $A_{j_1} = A_{j_2} = A_{j_3} \equiv A$ or (2) two $A_{j_\ell}$ are equal $(\equiv A)$ and one is not $(\equiv B)$.
    
    In the first case, Eq.~\eqref{TrA} becomes (written in the eigenvalues $\lambda_1,\lambda_2$ of $A$):
    \begin{equation}\label{Tr(A^3)}
    \begin{split}
            x_{(123)}^{( \, \vec{j} \, )} = -2\tr(A)^3+6 \tr(A)- 4 \tr(A^3) &= -2(\lambda_1 + \lambda_2)^3 + 6 (\lambda_1+ \lambda_2) - 4 (\lambda_1^3 + \lambda_2^3) \\
            &= -6(\lambda_1^2+\lambda_2^2)(\lambda_1+\lambda_2) + 6 (\lambda_1+\lambda_2) = 0,
    \end{split}
    \end{equation}
    since $\lambda_1^2 + \lambda_2^2 = \tr(A^2) = 1$.
    
    In the second case Equation (\ref{TrA}) becomes:
    \begin{equation*}
        x_{(123)}^{( \, \vec{j} \, )} = -2 \tr(A)^2 \tr(B) + 2 \tr(B) - 4 \tr(A^2 B).
    \end{equation*}
    Since the trace is invariant under choice of a basis, taking the trace in $A$'s eigenbasis leads to:
    \begin{equation*}
        x_{(123)}^{( \, \vec{j} \, )} = 2 \tr(B) (1-(\lambda_1+\lambda_2)^2) - 4 \tr(\left(\begin{array} {cc}
        \lambda_1^2 & 0\\
        0 & \lambda_2^2\\
        \end{array} \right)B) =  2 \tr(B) (2\lambda_1\lambda_2) - 4 (\lambda_1^2 B_{11} + \lambda_2^2 B_{22}).
    \end{equation*}
    Orthonormality provides $0 = \tr(AB) = \lambda_1 B_{11} + \lambda_2 B_{22}$, therefore $\lambda_1^2 B_{11} = -\lambda_1\lambda_2 B_{22}$ and $\lambda_2^2 B_{22} = -\lambda_1\lambda_2 B_{11}$ and thus
    \begin{equation*}
        x_{(123)}^{( \, \vec{j} \, )} =  -4 (B_{11}+B_{22}) \lambda_1 \lambda_2 - 4 (\lambda_1^2 B_{11} + \lambda_2^2 B_{22}) = -4 (B_{11}+B_{22}) \lambda_1 \lambda_2 -4 \cdot (-1) \lambda_1 \lambda_2 (B_{22} + B_{11}) = 0.
    \end{equation*}
\end{proof}

\begin{theorem}[Classification of all $r=3$ observables, where a measurement of third moments includes the determinant]\label{class0}
    Let ${\mathcal{O} = \sum\limits_{j=1}^3 s_j A_j \otimes B_j\in L(\mathbb{C}^2)^{\otimes 2}_{\mathrm{herm.}}}$ be an arbitrary observable of operator Schmidt rank $3$ in its operator Schmidt decomposition. Then $\mathcal{R}^{(3)}_{\mathcal{O}}(\varrho)$ is an affine linear combination of the invariants $\{|\vec{\alpha}|^2,|\vec{\beta}|^2,\tr(T^{\T} T),\langle\vec{\alpha},T\vec{\beta}\rangle,\det(T)\}$.
    Furthermore, consider the orthogonal projection mapping
    \begin{equation}
        \hat{P} : \mathbb{C}^{2\times2}_{\mathrm{herm.}} \twoheadrightarrow i \mathfrak{su}(2) = \{A \in L\left(\mathbb{C}^2\right)_{\mathrm{herm.}}  | \tr(A) = 0\} = \mathrm{span}_{\mathbb{R}}(\{\sigma_1,\sigma_2,\sigma_3\}), \quad A\mapsto \frac{1}{2} \sum\limits_{j=1}^3 \tr(A \sigma_j) \sigma_j.
    \end{equation}
    Then the prefactor in front of $\det(T)$ is $\neq 0$ if and only if 
    \begin{equation}
        \mathrm{span}(\hat{P}(A_1),\hat{P}(A_2),\hat{P}(A_3)) = \mathrm{span}(\hat{P}(B_1),\hat{P}(B_2),\hat{P}(B_3)) = i \mathfrak{su}(2).
    \end{equation}
    Consider the matrix $(M_{j,k})_{j,k=1,2,3} := \tr(\hat{P}(A_j)\hat{P}(B_k))$.
    Then the prefactor in front of $\det(T)$ is given by 
    \begin{equation}
         \frac{s_1 s_2 s_3}{8} \det(M).
    \end{equation}
\end{theorem}

\begin{proof}
    Due to Lemma \ref{t3d2 system}, the only term that can contribute to $\det(T)$ in $\tr(\varrho^{\otimes 3} \Phi^{(t)}(\mathcal{O}))$ is given by $\tr(\varrho^{\otimes 3} V_{(123),(123)}) = ... - \frac{3}{8} \det(T)$. Therefore, for an arbitrary operator $\mathcal{O}$, the prefactor in front of $\det(T)$ is given by 
    \begin{equation*}
        p \equiv -\frac{3}{8} \cdot c_{(123),(123)} = - \frac{3}{8} \sum_{j_1,j_2,j_3=1}^r s_{j_1}s_{j_2}s_{j_3} x_{(123)}^{( \, \vec{j} \, )} y_{(123)}^{( \, \vec{j} \,)}.
    \end{equation*}
    From the proof of Theorem \ref{r>2}, it is known that $x_{(123)}^{( \, \vec{j} \, )} = 0$ as soon as two entries in $\vec{j}$ coincide. Therefore,
    \begin{equation*}
        p = - \frac{3}{8} \sum_{j_1,j_2,j_3=1}^3 s_{j_1}s_{j_2}s_{j_3} x_{(123)}^{( \, \vec{j} \, )} y_{(123)}^{( \, \vec{j} \,)} = - \frac{3}{8} s_1 s_2 s_3 \sum_{\pi \in S_3} x_{(123)}^{(\pi(1),\pi(2),\pi(3))} y_{(123)}^{(\pi(1),\pi(2),\pi(3))}.
    \end{equation*}
    An expansion of $A_j \in \mathbb{C}_{\mathrm{herm.}}^{2\times 2}$ into the orthonormal basis $\left\{\frac{\mathbbm{1}}{\sqrt{2}},\frac{\sigma_x}{\sqrt{2}},\frac{\sigma_y}{\sqrt{2}},\frac{\sigma_z}{\sqrt{2}}\right\}$ is given by
    \begin{equation*}
        A_j = \frac{\tr(A_j)}{\sqrt{2}} \frac{\mathbbm{1}}{\sqrt{2}} + \sum\limits_{k=1}^3 \alpha_k^{(j)} \frac{\sigma_k}{\sqrt{2}}.
    \end{equation*}
    The orthogonality of $A_j$ and $A_k$ for $j \neq k$ leads to
    \begin{equation}
        0 = \tr(A_jA_k) = \frac{\tr(A_j)\tr(A_k)}{2} + \langle \overrightarrow{\alpha^{(j)}} , \overrightarrow{\alpha^{(k)}} \rangle.
    \end{equation}
    Equation (\ref{TrA}) becomes
    \begin{equation*}
        x_{(123)}^{(1,2,3)} = - \frac{1}{6} \tr(A_1) \tr(A_2) \tr(A_3) - \frac{1}{3} \tr(A_1A_2A_3).
    \end{equation*}
    The second term can be calculated explicitly:
    \begin{equation*}
    \begin{split}
        \tr(A_1A_2A_3) =& \tr(\left(\frac{\tr(A_1)}{\sqrt{2}} \frac{\mathbbm{1}}{\sqrt{2}} +\overrightarrow{\alpha^{(1)}} \cdot \frac{\vec{\sigma}}{\sqrt{2}} \right) \left(\frac{\tr(A_2)}{\sqrt{2}} \frac{\mathbbm{1}}{\sqrt{2}} +\overrightarrow{\alpha^{(2)}} \cdot \frac{\vec{\sigma}}{\sqrt{2}} \right)\left(\frac{\tr(A_3)}{\sqrt{2}} \frac{\mathbbm{1}}{\sqrt{2}} +\overrightarrow{\alpha^{(3)}} \cdot \frac{\vec{\sigma}}{\sqrt{2}} \right)) \\
        =& \frac{\tr(A_1) \tr(A_2) \tr(A_3)}{\sqrt{2}^3}\cdot \frac{\tr(\mathbbm{1}^3)}{\sqrt{2}^3} + \frac{\tr(A_1)}{2} \langle \overrightarrow{\alpha^{(2)}}, \overrightarrow{\alpha^{(3)}} \rangle + \frac{\tr(A_2)}{2} \langle \overrightarrow{\alpha^{(1)}}, \overrightarrow{\alpha^{(3)}} \rangle \\
        &+ \frac{\tr(A_3)}{2} \langle \overrightarrow{\alpha^{(1)}}, \overrightarrow{\alpha^{(2)}} \rangle + \sum_{j,k,\ell=1}^3 \alpha^{(1)}_j\alpha^{(2)}_k\alpha^{(3)}_\ell \frac{1}{\sqrt{2}^3} \tr(\sigma_j \sigma_k \sigma_\ell).
    \end{split}
    \end{equation*}

    Notice that 
    \begin{equation*}
        \sum_{j,k,\ell=1}^3 \alpha^{(1)}_j\alpha^{(2)}_k\alpha^{(3)}_\ell \tr(\sigma_j \sigma_k \sigma_\ell) = 2i \sum_{j,k,\ell=1}^3 \alpha^{(1)}_j\alpha^{(2)}_k\alpha^{(3)}_\ell \frac{1}{\sqrt{2}^3} \varepsilon_{jk\ell} = 2i \det(\overrightarrow{\alpha^{(1)}},\overrightarrow{\alpha^{(2)}},\overrightarrow{\alpha^{(3)}}),
    \end{equation*}
    and that by orthogonality $\frac{\tr(A_1)}{2} \langle \overrightarrow{\alpha^{(2)}}, \overrightarrow{\alpha^{(3)}} \rangle = - \frac{\tr(A_1) \tr(A_2)\tr(A_3)}{4}$.
    Thus
    \begin{equation*}
    \begin{split}
        &\tr(A_1A_2A_3) = \frac{\tr(A_1)\tr(A_2)\tr(A_3)}{4} - 3\cdot \frac{\tr(A_1)\tr(A_2)\tr(A_3)}{4} + \frac{i}{\sqrt{2}} \det(\overrightarrow{\alpha^{(1)}},\overrightarrow{\alpha^{(2)}},\overrightarrow{\alpha^{(3)}})\\
        &\Rightarrow x_{(123)}^{(1,2,3)} = - \frac{1}{6} \tr(A_1)\tr(A_2)\tr(A_3) - \frac{1}{3} \left(-\frac{1}{2}\right) \tr(A_1)\tr(A_2)\tr(A_3) - \frac{i}{3\sqrt{2}} \det(\overrightarrow{\alpha^{(1)}},\overrightarrow{\alpha^{(2)}},\overrightarrow{\alpha^{(3)}})\\
        & \qquad  \qquad \; = - \frac{i}{3\sqrt{2}} \det(\overrightarrow{\alpha^{(1)}},\overrightarrow{\alpha^{(2)}},\overrightarrow{\alpha^{(3)}}).
    \end{split}
    \end{equation*}
    Therefore,
    \begin{equation*}
    \begin{split}
        p &= - \frac{3}{8} s_1 s_2 s_3 \sum_{\pi \in S_3} x_{(123)}^{(\pi(1),\pi(2),\pi(3))} y_{(123)}^{(\pi(1),\pi(2),\pi(3))} \\
        &= - \frac{3}{8} s_1 s_2 s_3 \left( \frac{i}{3\sqrt{2}} \right)^2 \sum_{\pi \in S_3} \det(\overrightarrow{\alpha^{(\pi(1))}},\overrightarrow{\alpha^{(\pi(2))}},\overrightarrow{\alpha^{(\pi(3))}})\det(\overrightarrow{\beta^{(\pi(1))}},\overrightarrow{\beta^{(\pi(2))}},\overrightarrow{\beta^{(\pi(3))}}) \\
        &= \frac{s_1s_2s_3}{8\cdot 3 \cdot 2}\cdot 3! \det(\overrightarrow{\alpha^{(1)}},\overrightarrow{\alpha^{(2)}},\overrightarrow{\alpha^{(3)}})\det(\overrightarrow{\beta^{(1)}},\overrightarrow{\beta^{(2)}},\overrightarrow{\beta^{(3)}}).
    \end{split}
    \end{equation*}
    Notice that this is precisely $\neq 0$ if $\det(\overrightarrow{\alpha^{(1)}},\overrightarrow{\alpha^{(2)}},\overrightarrow{\alpha^{(3)}})$ and $\det(\overrightarrow{\beta^{(1)}},\overrightarrow{\beta^{(2)}},\overrightarrow{\beta^{(3)}})$ are $\neq 0$. Since these are the expansion coefficients into an orthonormal basis of $i\mathfrak{su}(2)$ this is equivalent to 
    
    $\mathrm{span}(\hat{P}(A_1),\hat{P}(A_2),\hat{P}(A_3)) = \mathrm{span}(\hat{P}(B_1),\hat{P}(B_2),\hat{P}(B_3)) = i \mathfrak{su}(2)$.
    Using $\det(A) = \det(A^{\T} )$ and 
    
    $\det(A^{\T} )\det(B) = \det(A^{\T} B) = \det(\langle \vec{a}^{(i)}, \vec{b}^{(j)} \rangle)$ leads to 
    \begin{equation}\label{Matrixprod}
        k = \frac{s_1s_2s_3}{8} \det(\left( \overrightarrow{\alpha^{(1)}},\overrightarrow{\alpha^{(2)}},\overrightarrow{\alpha^{(3)}}\right)^{\T} ) \det(\overrightarrow{\beta^{(1)}},\overrightarrow{\beta^{(2)}},\overrightarrow{\beta^{(3)}}) = \frac{s_1s_2s_3}{8} \det( \left(\langle \alpha^{(j)}, \beta^{(k)} \rangle \right)_{j,k})
    \end{equation}
    and $\tr(\hat{P}(A_j)\hat{P}(B_k)) = \tr(\left(\vec{\alpha^{(j)}} \cdot \frac{\vec{\sigma}}{\sqrt{2}} \right) \left( \vec{\beta^{(k)}} \cdot \frac{\vec{\sigma}}{\sqrt{2}}\right)) = \langle \alpha^{(j)}, \beta^{(k)} \rangle$.
\end{proof}

\begin{examp}\label{pauli examp}
    Let $\mathcal{O} = \sum\limits_{j=1}^3 \sigma_j \otimes \sigma_j$. Then  $A_j = B_j = \hat{P}(A_j) = \hat{P}(B_j) = \frac{\sigma_j}{\sqrt{2}}$ and $s_j = 2$. This leads to 

    $M_{jk} = \tr(\frac{\sigma_j}{\sqrt{2}}\frac{\sigma_k}{\sqrt{2}}) = \delta_{ik}$ and thus $\det(M) = 1$ and $\frac{s_1 s_2 s_3}{8} \det(M) = \frac{2^3}{8} \cdot 1 = 1$. Therefore, $\mathcal{R}^{(3)}_{\mathcal{O}} (\varrho) = ... + \det(T)$. 
\end{examp}

\begin{corollary}[Classification of all observables, where a measurement of third moments includes the determinant]\label{class}
    Let ${\mathcal{O} = \sum\limits_{j=1}^r s_j A_j \otimes B_j\in L(\mathbb{C}^2)^{\otimes 2}_{\mathrm{herm.}}}$ be an arbitrary observable in its operator Schmidt decomposition. Then $\mathcal{R}_{\mathcal{O}}^{(3)}(\varrho)$ is an affine linear combination of the invariants $\{|\vec{\alpha}|^2,|\vec{\beta}|^2,\tr(T^{\T} T),\langle\vec{\alpha},T\vec{\beta}\rangle,\det(T)\}$.
    Furthermore, let $\tilde{P} : \mathbb{C}^{2\times2}_{\mathrm{herm.}} \twoheadrightarrow i \mathfrak{su}(2) \cong \mathbb{R}^3$ be the orthogonal projection mapping given by $A\mapsto \left(\tr(A \frac{\sigma_j}{\sqrt{2}})\right)_j$.
    Let 
    \begin{equation}
        M_A := \left(\sqrt{s_1}\tilde{P}(A_1),...,\sqrt{s_r} \tilde{P}(A_r)\right) \in \mathbb{R}^{3\times r}
    \end{equation}
    and define $M_B$ analogously.
    Then the prefactor in front of $\det(T)$ in $\mathcal{R}^{(3)}_{\mathcal{O}}(\varrho)$ is given by $\frac{\det(M_A M_B^{\T} )}{8}$. It is non-zero if and only if $M_AM_B^{\T} $ is invertible.
\end{corollary}
\begin{proof}
    For $r=1,2$, the prefactor is always zero, since $\mathrm{rank}(M_B) \leq 2$ and thus the $3\times 3$ matrix $M_AM_B^{\T} $ cannot be invertible. And indeed, according to Theorem \ref{r>2}, that is the relevant result here.
    
    The $r=3$ case is similar to Theorem~\ref{class0} but slightly reframed. More precisely, it follows from Eq.~\eqref{Matrixprod} by distributing the $s_i$ equally into both matrices. Notice that the remaining $r=4$ case can be discussed in the same way as the $r=3$ case. This leads to the prefactor
    \begin{equation*}
        p = \frac{1}{8} \sum_{\substack{(j,k,\ell) \in \\ \{(1,2,3),(1,2,4),\\(1,3,4),(2,3,4)\}}} s_is_js_{\ell} \det(\overrightarrow{\alpha^{(j)}},\overrightarrow{\alpha^{(k)}},\overrightarrow{\alpha^{(\ell)}}) \det(\overrightarrow{\beta^{(j)}},\overrightarrow{\beta^{(k)}},\overrightarrow{\beta^{(\ell)}}).
    \end{equation*}
    The claim now follows from the Cauchy–Binet formula.
\end{proof}

This statement fully classifies all observables that have a dependence on the determinant. Since any $t=3$ measurement can otherwise only include type 1 invariants, this means that $\det(T)$ is precisely obtainable through these observables up to a type 1 measurement.

\section{Measuring only $\det(T)$}\label{app:onlydetT}
\begin{proposition}\label{det op} Let $\mathcal{O} = \sum\limits_{j=1}^3 \sigma_j \otimes \sigma_j$ be the observable from Example \ref{pauli examp} and let $\varrho$ be an arbitrary bipartite qubit state. Then $\mathcal{R}^{(3)}_{\mathcal{O}}(\varrho) = \det(T)$.
\end{proposition}
\begin{proof}
    \begin{equation*}
    \begin{split}
         \Phi^{(3)}\left(\mathcal{O}\right)& = \iint (U_A^\dagger \otimes U_B^\dagger)^{\otimes 3} \left(\sum\limits_{j=1}^3 \sigma_j \otimes \sigma_j\right)^{\otimes 3} (U_A \otimes U_B)^{\otimes 3} \, \dd U_A \dd U_B \\
        &=\iint (U_A^\dagger \otimes U_B^\dagger)^{\otimes 3} \Biggl( \sum_{j=1}^3 \sigma_j^{\otimes 6} + \sum_{\substack{j,k=1 \\ k \neq j}}^3 \left[ \sigma_j^{\otimes 4}\otimes \sigma_k^{\otimes 2} + \sigma_k^{\otimes 2} \sigma_j^{\otimes 2} \sigma_k^{\otimes 2} + \sigma_k^{\otimes 2}\sigma_j^{\otimes 4} \right] \\
        & \qquad \qquad\qquad \qquad \quad \;+\sum_{\pi \in S_3} \sigma_{\pi(1)}^{\otimes 2} \otimes \sigma_{\pi(2)}^{\otimes 2} \otimes \sigma_{\pi(3)}^{\otimes 2}\Biggr)  (U_A \otimes U_B)^{\otimes 3}  \, \dd U_A \dd U_B.
    \end{split}
    \end{equation*}
    
    Three kinds of terms appear in the integral. Concerning the first kind:
    \begin{equation}
        \iint (U_A^\dagger \otimes U_B^\dagger)^{\otimes 3} \sigma_j^{\otimes 6}(U_A \otimes U_B)^{\otimes 3} \, \dd U_A \dd U_B = \left(\int U^{\dagger \, \otimes 3} \sigma_j^{\otimes 3} U^{\otimes 3} \dd U \right)^{\otimes 2} \equiv \mathcal{N}^{\otimes 2}_1,
    \end{equation}
    with $\mathcal{N}_1 = \int U^{\dagger \, \otimes 3} \sigma_j^{\otimes 3} U^{\otimes 3} \dd U = \sum\limits_{\pi \in S_3} c_\pi V_\pi$.
    
    We have that $\tr(\mathcal{N}_1 V_\pi) = \tr(\sigma_j^{\otimes 3} V_\pi) = 0$ for all $\pi \in S_t$, since ${\tr(\sigma_j) = \tr(\sigma_j^3) = 0}$ and only the terms  $\tr(\sigma_j)^3, \tr(\sigma_j)\tr(\sigma_j^2)$ and $\tr(\sigma_j^3)$ can appear via Lemma \ref{calc}. Lemma \ref{sym} now implies $\mathcal{N}_1 = 0$.

    For the second kind, we have
    \begin{equation}
        \mathcal{N}_2 = \int U^{\dagger \, \otimes 3} \left( \sigma_j^{\otimes 2} \otimes \sigma_k \right) U^{\otimes 3} \dd U,
    \end{equation}
    and $\tr(\mathcal{N}_2 V_\pi) = 0$ for all $\pi \in S_3$, since only the terms 
    \begin{equation*}
        \tr(\sigma_j^2) \tr(\sigma_k), \tr(\sigma_j)^2 \tr(\sigma_k), \tr(\sigma_j\sigma_k) \tr(\sigma_j), \tr(\sigma_j \sigma_k \sigma_j)=\tr(\sigma_j^2 \sigma_k) = \tr(\sigma_k)
    \end{equation*}
    can appear and are all zero, again implying $\mathcal{N}_2 = 0$. All other permutations in the second sum yield the same result.

    For the third kind, we find for $\pi \in S_t$ that
    \begin{equation}
        \mathcal{N}_3 = \int U^{\dagger \, \otimes 3} \left(\sigma_{\pi(1)} \otimes \sigma_{\pi(2)} \otimes \sigma_{\pi(3)} \right) U^{\otimes 3} \dd U.
    \end{equation}
    Note that $\tr(\mathcal{N}_3 \mathbbm{1}) = \tr(\sigma_x) \tr(\sigma_y) \tr(\sigma_z) = 0$ and $\tr(\mathcal{N}_3 V_{(jk)}) = \tr(\sigma_j \sigma_k) \tr(\sigma_\ell) = 0$ (where $\{j,k,\ell\} = \{1,2,3\}$).
    Therefore, only the $3$-cycles will induce a non-zero trace:

    \begin{equation*}
    \begin{split}
    \tr(\mathcal{N} V_{(123)}) &= \tr(\sigma_{\pi(3)} \sigma_{\pi(1)} \sigma_{\pi(2)})\\
    &= 
        \begin{cases}
            \; \; \; i \tr(\sigma_{\pi(2)}^2) = i \tr(\mathbbm{1}) &= \; \; \;2i \quad \text{if } \pi(3) < \pi(1) \pmod 3 \\
            -i \tr(\sigma_{\pi(2)}^2) &= -2i \quad \text{otherwise.}
        \end{cases}.
    \end{split}
    \end{equation*}

    Conversely, calculating $\tr(\mathcal{N}_3 V_\pi)$ in terms of $c_{\pi'}$ via $\mathcal{N}_3 = \sum\limits_{\pi' \in S_3} c_{\pi'} V_{\pi'}$ and evaluating $\tr(V_{\pi' \circ \pi})$ using Lemma \ref{calc} leads to

    \begin{equation}
    \left(\begin{array} {cccccc}
    8 & 4 & 4 & 4 & 2 & 2 \\
    4 & 8 & 2 & 2 & 4 & 4 \\
    4 & 2 & 8 & 2 & 4 & 4 \\
    4 & 2 & 2 & 8 & 4 & 4 \\
    2 & 4 & 4 & 4 & 2 & 8 \\
    2 & 4 & 4 & 4 & 8 & 2 \\
    \end{array} \right) \begin{bmatrix}
    c_{()}\\
    c_{(12)}\\
    c_{(13)}\\
    c_{(23)}\\
    c_{(123)}\\
    c_{(132)}\\
    \end{bmatrix} = \begin{bmatrix}
    0\\
    0\\
    0\\
    0\\
    \pm 2i\\
    \mp 2i\\
    \end{bmatrix},
    \end{equation}

which is for example solved by $\left(c_{()},c_{(12)},c_{(13)},c_{(23)},c_{(123)},c_{(132)} \right)^{\T}  = \left(0,0,0,0, \mp\frac{i}{3}, \pm \frac{i}{3}\right)^{\T} $.

Therefore,
\begin{equation}
    \Phi^{(3)} (\sigma_x \otimes \sigma_x + \sigma_y \otimes \sigma_y + \sigma_z \otimes \sigma_z) = \frac{6i^2}{3^2} \left(V_{(123)} - V_{(132)}\right)^{\otimes 2}  = - \frac{2}{3} \left(V_{(123)} - V_{(132)}\right)^{\otimes 2}.
\end{equation}

Finally, consider $\tr(\varrho^{\otimes 3} \Phi^{(3)}(\mathcal{M})) = - \frac{2}{3} \tr(\varrho^{\otimes 3} \left(V_{(123)} - V_{(132)}\right)^{\otimes 2})$. Expanding $\varrho$ leads to a sum of terms of the form 
\begin{equation*}
    \tr(A_1 \otimes A_2 \otimes A_3 (V_{(123)} -V_{(132)}))\tr(A_4 \otimes A_5 \otimes A_6 (V_{(123)} -V_{(132)})),
\end{equation*}
where each $A_j$ can be a Pauli matrix or an identity. Notice that
\begin{equation*}
    \tr(A_1 \otimes A_2 \otimes A_3 (V_{(123)} -V_{(132)})) = \tr(A_1A_2A_3) - \tr(A_1A_3A_2),
\end{equation*}
and if at least one of them is an identity, both traces are the same due to the cyclic property.
Thus the only terms remaining have the form 
\begin{equation*}
    \tr(\sigma_j \sigma_k \sigma_\ell) - \tr(\sigma_j \sigma_\ell \sigma_k) = \tr(\sigma_j [\sigma_k,\sigma_\ell]) = 2i \varepsilon_{k\ell m} \tr(\sigma_j \sigma_m) =4i \varepsilon_{k\ell j} = 4i \varepsilon_{jk\ell}.
\end{equation*}

Thus, 
\begin{equation*}
\begin{split}
    \tr(\varrho^{\otimes 3} \Phi^{(3)}(\mathcal{M})) &= - \frac{2}{3} \tr(\varrho^{\otimes 3} \left(V_{(123)} - V_{(132)}\right)^{\otimes 2}) = - \frac{2}{3\cdot 4^3} (4i)^2 \sum_{\substack{j_1,k_1,j_2, \\ k_2,j_3,k_3 = 1}}^3 T_{j_1,k_1} T_{j_2,k_2} T_{j_3,k_3} \varepsilon_{j_1j_2j_3} \varepsilon_{k_1k_2k_3} \\
    &= \frac{1}{6} \cdot 3! \det(T) = \det(T),
\end{split}
\end{equation*}
by the Leibniz formula.    
\end{proof}

\begin{proposition}\label{class 1}
    Let $(A_j)_j$ and $(B_j)_j$ be orthonormal bases of $i \mathfrak{su}(2)$ and $(s_j)_j$ real numbers. Let $\mathcal{O} = \sum\limits_{j=1}^3 s_j A_j \otimes B_j$ and let $\varrho$ be an arbitrary bipartite qubit state. Then $\mathcal{R}_{\mathcal{O}}^{(3)}(\varrho) = \pm \frac{s_1s_2s_3}{8} \det(T)$, where the prefactor $\pm$ is the relative orientation of $A$ to $B$.
\end{proposition}
\begin{proof}
    Assume that both $(A_j)_j$ and $(B_j)_j$ are positively oriented (otherwise $(-A_j)_j$ is positively oriented, and $\mathcal{O} \to - \mathcal{O}$ leads to $\mathcal{R}^{(3)} \to (-1)^3\mathcal{R}^{(3)}$. Then by the double covering, there exist unitary maps $U_A,U_B$ such that $U_A^{\dagger} A_j U_A = \frac{\sigma_j}{\sqrt{2}}$ and $U_B^\dagger B_j U_B = \frac{\sigma_j}{\sqrt{2}}$. Since $\mathcal{R}$ is invariant under $\mathcal{O} \to (U_A \otimes U_B)^{\dagger} \mathcal{O} (U_A \otimes U_B)$, it is sufficient to consider $\mathcal{O} = \sum\limits_{j=1}^3 \frac{s_j}{\sqrt{2}^2} \sigma_j \otimes \sigma_j$. In this case the statement requires an easy modification of the proof of Proposition \ref{det op}, since only $\mathcal{N}_3^{\otimes 2} \neq 0$ and its prefactor is now $\frac{s_1s_2s_3}{\sqrt{2}^6}$.
\end{proof}

An explicit calculation allows the proof of the other direction in Proposition \ref{class 1}, but only for symmetric measurements using the following Lemma, as well as the symmetric version of Lemma \ref{t3d2 system}.
\begin{lemma}
\label{d=2 sym theorem}
Let $\mathcal{O} = \sum\limits_{j=1}^r s_j A_j \otimes A_j\in L(\mathbb{C}^2)^{\otimes 2}_{\mathrm{herm.}}$ be an observable with a symmetric operator Schmidt decomposition. 

For $j \in \mathbb{N}$ let $\vec{\tau^j} = (\tr(A_1)^j,...,\tr(A_r)^j)^{\T} $ and $\vec{s^j} = (s_1^j,...,s_r^j)^{\T} $. Then the prefactors $(c_{\pi,\pi'})_{\pi,\pi' \in S_3}$ from Proposition \ref{explicit} for $t=3$ are given by
\begin{equation}\label{c Mat}
\underbrace{\begin{bmatrix*}[c]
c_{(),()} \\[7pt]
c_{(),(ij)} \\[7pt]
c_{(),(123)} \\[7pt]
c_{(ij),(ij)} \\[7pt]
c_{(ij),(ik) \; j \neq k} \\[7pt]
c_{(ij),(123)} \\[7pt]
c_{(123),(123)} \\[7pt]
\end{bmatrix*}}_{\tilde{c}(\mathcal{O})} = \underbrace{\frac{1}{144}
    \left(\begin{array} {ccccccc} 
    4 & 0 & -8 & 0 & 0 & 0 & 4\\[7pt]
    0 & -2 & 4 & 0 & 0 & 2 & -4\\[7pt]
    -4 & 12 & -4 & 0 & 0 & -12 & 8 \\[7pt]
    0 & 0 & 0 & 3 & -2 & -4 & 4 \\[7pt]
    0 & 0 & 0 & -1 & 2 & -4 & 4 \\[7pt]
    0 & 2 & -4 & -2 & -4 & 16 & -8 \\[7pt]
    4 & -24 & 16 & 12 & 24 & -48 & 16\\ \end{array} \right)}_{\equiv B} 
    \underbrace{\begin{bmatrix*}[c]
    \langle \, \vec{s}, \vec{\tau^2} \rangle^3\\[3pt]
    \langle \vec{s^2},\vec{\tau^2} \rangle \langle \vec{s},\vec{\tau^2} \rangle\\[3pt]
    \sum\limits_{j_1,j_2,j_3=1}^r s_{j_1}\tau_{j_1}s_{j_2}\tau_{j_2}s_{j_3}\tau_{j_3}\tr(A_{j_1}A_{j_2}A_{j_3})\\[5pt]
    |\vec{s}|^2 \langle \vec{s},\vec{\tau^2} \rangle\\[3pt]
    \langle \vec{s^3},\vec{\tau^2} \rangle \\[3pt]
     \sum\limits_{j,k=1}^r s_j^2\tau_k s_k \tr(A_j^2A_k) \\[5pt]
    \sum\limits_{j_1,j_2,j_3=1}^r s_{j_1}s_{j_2}s_{j_3} \tr(A_{j_1}A_{j_2}A_{j_3})^2 
    \end{bmatrix*}}_{\equiv \tilde{v}(\mathcal{O})}
\end{equation}
    
\end{lemma}

\begin{proposition}\label{strong char twirl}
    Let $a,b \in \mathbb{C}, b\neq 0$. The set of observables $\mathcal{O}\in L(\mathbb{C}^2)^{\otimes 2}_{\mathrm{herm.}}$ with a symmetric operator Schmidt decomposition that satisfy
    $\mathcal{R}^{(t)}_{\mathcal{O}}(\varrho) =  a + b \det(T)$ for arbitrary bipartite qubit states $\varrho$ is only nonempty for $a = 0$ and $b \in \mathbb{R}_{>0}$. In that case it is given by:
    \begin{equation}
        \left\{ (U\otimes U)^\dagger \left(\sum_{j=1}^3 s_j \sigma_j \otimes \sigma_j \right) (U\otimes U) \; \bigg| \; U \in SU(2), \; s_1s_2s_3 = 8b, \; s_j>0 \; \forall \; j \right\}.
    \end{equation}
    Notably, there are no symmetric $r=4$ measurements, that only measure the determinant.
\end{proposition}
\begin{proof}
    By Lemma \ref{t3d2 system}, since $\mathcal{O}$ has a symmetric operator Schmidt decomposition, $\mathcal{R}_{\mathcal{O}}^{(t)}(\varrho) = \sum\limits_{j=1}^7 \tilde{c}_j(\mathcal{O}) f_j(\varrho)$ where $\tilde{c}(\mathcal{O})$ is given by Equation (\ref{c Mat}).
    To get the determinant without interference from other invariants, one needs to find an observable $\mathcal{O}$ such that 

\begin{equation}
    \mathcal{R}_{\mathcal{O}}^{(t)}(\varrho) = \sum_{j=1}^7 \tilde{c}_j(\mathcal{O}) f_j(\varrho) = \sum_{j,\ell = 1}^7 \tilde{c}_j(\mathcal{O})_j D_{j,\ell} \tilde{I}_\ell(\varrho) \stackrel{!}{=}  a \tilde{I}_1(\varrho) + b \tilde{I}_5(\varrho)
\end{equation}
for some $a,b \in \mathbb{C}$ with $b \neq 0$. This implies that $D^{\T}  \tilde{c}(\mathcal{O}) = (a,0,0,0,b)^{\T} $. All such solutions have the form of a special solution plus the nullspace of $D^{\T} $.
A special solution of $D^{\T}  \tilde{c}(\mathcal{O}) = (a,0,0,0,b)^{\T} $ is for example given by
$\left(a+\frac{2b}{3},-\frac{2b}{3},0,\frac{2b}{3},\frac{2b}{3},0,-\frac{8b}{3}\right)^{\T} $, while the nullspace of $D^{\T} $ can be calculated to be $\mathrm{span}\left(\left\{(2,-1,2,0,0,0,0)^{\T} ,(0,1,0,-2,-2,2,0)^{\T} \right\}\right)$.
Thus all possible $\tilde{c}(\mathcal{O})$ have the form
\begin{equation*}
        \tilde{c}(\mathcal{O}) = \begin{bmatrix*}[l]
        a + \frac{2}{3}b + 2x +0y\\[5pt]
        \quad -\frac{2}{3}b - \; \;x + \; \: y \\[5pt]
        \qquad \quad \quad 2x \\
        \qquad \frac{2}{3}b \qquad \; -2y \\[5pt]
        \qquad \frac{2}{3}b \qquad \; -2y \\[5pt]
        \qquad \qquad \quad \; \; \: +2y \\
        \quad - \frac{8}{3}b\\[5pt]
        \end{bmatrix*}
\end{equation*}
for some $a,b,x,y\in \mathbb{C}$.
Inverting Equation (\ref{c Mat}) leads to
\begin{equation*}
    \tilde{v}(\mathcal{O})\begin{bmatrix*}[c]
    \langle \, \vec{s}, \vec{\tau^2} \rangle^3\\[3pt]
    \langle \vec{s^2},\vec{\tau^2} \rangle \langle \vec{s},\vec{\tau^2} \rangle\\[3pt]
    \sum\limits_{j_1,j_2,j_3=1}^r s_{j_1}\tau_{j_1}s_{j_2}\tau_{j_2}s_{j_3}\tau_{j_3}\tr(A_{j_1}A_{j_2}A_{j_3})\\[5pt]
    |\vec{s}|^2 \langle \vec{s},\vec{\tau^2} \rangle\\[3pt]
    \langle \vec{s^3},\vec{\tau^2} \rangle \\[3pt]
     \sum\limits_{j,k=1}^r s_j^2\tau_k s_k \tr(A_j^2A_k) \\[5pt]
    \sum\limits_{j_1,j_2,j_3=1}^r s_{j_1}s_{j_2}s_{j_3} \tr(A_{j_1}A_{j_2}A_{j_3})^2 
    \end{bmatrix*} =B^{-1} \tilde{c}(\mathcal{O}) = \begin{bmatrix*}[l]
        64a \\
        32a \\
        16a+16b-48x-72y \\
        16a \\
        16a \\
        8a+32b-24x-72y \\
        4a+56b-24x-144y\\
        \end{bmatrix*}.
\end{equation*}

Subtracting the fifth from the fourth line implies that 
\begin{equation*}
    0 \stackrel{!}{=} |\vec{s}|^2 \langle \vec{s},\vec{\tau^2} \rangle -
    \langle \vec{s^3},\vec{\tau^2} \rangle = \sum_{j,k = 1}^r s_j^2 s_k \tau_k^2 - \sum_{j=1}^r s_k^3 \tau_k^2 = \sum_{k=1}^r \tau_k^2 s_k \sum_{\substack{j=1 \\ j \neq k}}^r s_j^2.
\end{equation*}
Since $s_j>0$ this implies that $r=1$ (which has already been decided to be impossible by Theorem \ref{r>2}) or that all $\tau_k \equiv 0$. This then implies (looking at line one, for example), that $a=0$. Furthermore, $\vec{\tau} \equiv 0$ leads to the situation in Proposition \ref{class 1}.
\end{proof}
Attacking the non-symmetric case in the same way is considerably harder, since the nullspace of $D^{\T} $ is larger, and terms of the form $\tau_j^2$ are sometimes replaced by $\tr(A_j)\tr(B_j)$ loosing the obvious positivity condition, which was essential for the proof.

\section{$t=4$ and the Hodge invariant}\label{app:t4hodge}
The analysis for the Hodge invariant $I_{14}$ can be done mostly along the lines of the analysis of determinant $I_1$, with the biggest difference being that the $6\times 6$ matrix $M_2^{(3)}$ (which using symmetries can be reduced to an invertible $5\times 5$ matrix) is now replaced by the $24 \times 24$ Matrix $M_2^{(4)}$, which becomes an invertible $14 \times 14$ matrix after implementing the symmetries from Lemma~\ref{sym}.

The invariant in question can explicitly be written as

\begin{equation}\label{Hodge exp}
    I_{14} = 2\langle \vec{\alpha}, \mathrm{adj}(T)\vec{\beta} \rangle =  - \sum_{i,j,k,\ell,m,n=1}^3 \varepsilon_{ijk}\alpha_k T_{j,\ell} \varepsilon_{\ell m n} \beta_n T_{im} = \tr((*\vec{\alpha})T (*\vec{\beta})^{\T}  T^{\T} ),
\end{equation}
where $\mathrm{adj}(T) \in \mathbb{R}^{3 \times 3}$ is the adjugate of $T \in \mathbb{R}^{3\times 3}$ and $*\vec{v}$ is defined via $(*\vec{v})\vec{w} = \vec{v} \times \vec{w}$ for any $\vec{v},\vec{w} \in \mathbb{R}^3$.

It needs to be measured via $\mathcal{R}^{(4)}_{\mathcal{O}} (\varrho) = \tr(\varrho^{\otimes 4} \Phi^{(4)}(\mathcal{O})) = \sum\limits_{\pi,\pi' \in S_4} c_{\pi,\pi'} \tr(\varrho^{\otimes 4}V_{\pi,\pi'})$ which uses
\begin{equation}
    \varrho^{\otimes 4} = \frac{1}{2^8} \left[ \mathbbm{1} \otimes \mathbbm{1} + \sum_{k=1}^3 \alpha_k \sigma_k \otimes \mathbbm{1} + \sum_{k=1}^3\beta_k \mathbbm{1} \otimes \sigma_k + \sum_{j,k=1}^3 T_{jk} \sigma_j \otimes \sigma_k \right]^{\otimes 4}.
\end{equation}
From the explicit formula (\ref{Hodge exp}) it is clear, that $I_{14}$ arises from the terms in the expansion of $\varrho^{\otimes 4}$ that include $\vec{\alpha}\cdot \vec{\sigma} \otimes \mathbbm{1}$ and $\vec{\beta} \cdot \mathbbm{1} \otimes \vec{\sigma}$ once, and $\sum\limits_{jk} T_{jk} \sigma_j \otimes \sigma_k$ twice.
Similarly to the determinant (see Lemma \ref{t3d2 system}), one can now reason via Lemma \ref{calc} for which permutations
$(\pi_A,\pi_B) \in S_4^{2}$ the expression $\tr(\varrho^{\otimes 4}V_{\pi_A,\pi_B})$ could possibly include a dependence on $I_{14}$.

Notice that $V_{\pi_A,()}$ and $V_{(),\pi_B}$ can directly be excluded, since $()$ induces a product of all traces in one coordinate, and there is always a Pauli matrix present.
More generally, $\pi_A$ and $\pi_B$ cannot have more than one fixpoint, since in every coordinate there are $3$ Pauli matrices and one identity.
If $\pi_A$ has a fixpoint, then $\pi_B$ cannot be the same fixpoint, since this would imply a term of the form $\mathbbm{1} \otimes \mathbbm{1}$, which isn't present. The following possibilities remain:

\begin{enumerate}[label={(\arabic*)}]
    \item $\pi_A$ and $\pi_B$ both have exactly one fixpoint, and it is not the same one,
    \item  $\pi_A$ has one fixpoint and $\pi_B$ does not,
    \item $\pi_B$ has one fixpoint and $\pi_A$ does not,
    \item Neither $\pi_A$ nor $\pi_B$ have a fixpoint.
\end{enumerate}
Starting with case (4), it can be shown that it does not contribute a term to $I_{14}$. To see this, expand $\varrho^{\otimes 4}$ in $\tr(\varrho^{\otimes 4}V_{\pi_A,\pi_B})$. Clearly, $I_{14}$ can only result from the term
\begin{equation*}
    \sum_{i,j,k,\ell,m,n} \alpha_i \beta_j T_{k\ell}T_{mn} \tr(V_{\pi_A,\pi_B} \sum_{\mathrm{perm.}} (\sigma_i \otimes \mathbbm{1})\otimes(\mathbbm{1}\otimes \sigma_j) \otimes (\sigma_{k}\otimes \sigma_\ell)\otimes (\sigma_m \otimes \sigma_n)),
\end{equation*}
where the inner sum spans over all ways to permute the four bracketed terms. W.l.o.g., let $\pi_A = (1234)$. Explicitly evaluating the resulting trace leads to multiple terms of the form $\varepsilon_{ikm}\cdot \left(\pm \varepsilon_{j\ell n}\right) + \text{perm. }$, 
where the second sign depends on $\pi_B$. Independent of the precise form of $\pi_B$, its effect on 
$(\sigma_i \otimes \mathbbm{1})\otimes(\mathbbm{1}\otimes \sigma_j) \otimes (\sigma_{k}\otimes \sigma_\ell)\otimes (\sigma_m \otimes \sigma_n)$ and on $(\sigma_i \otimes \mathbbm{1})\otimes(\sigma_{k}\otimes \sigma_\ell) \otimes(\mathbbm{1}\otimes \sigma_j) \otimes (\sigma_m \otimes \sigma_n)$
differs by a sign, while for $(1234)$ it is identical. Thus, even though $I_{14}$ appears multiple times in the expansion of $\tr(\varrho^{\otimes 4} V_{\pi_A,\pi_B})$, due to the alternating prefactors, they all sum to zero.
This offers the first point of attack:
\begin{lemma}\label{suf}
    Showing that the Hodge invariant $I_{14}$ cannot be obtained through a measurement of fourth moments using observables of operator Schmidt rank $r\leq3$ is equivalent to showing that $c_{\pi,\pi'}$ always vanishes if $\pi$ or $\pi'$ has exactly one fixpoint. It is thus sufficient to show that $x^{( \; \vec{j} \;)}_{(ijk)}$ vanishes for all $3$-cycles in $S_4$ and all $\vec{j} \in \{1,2,3\}^4$, where $x_\pi^{( \, \vec{j} \, )}$ is defined via Proposition \ref{explicit}.
\end{lemma}
Using the symmetries from Lemma \ref{sym} the prefactors can be chosen, such that only $c_{(123)}^{ \, (\,\vec{j} \,)}$ is possibly non-zero. Explicitly inverting the resulting matrix yields:
\begin{equation}\label{important}
\begin{split}
    6 x^{( \, \vec{j} \, )}_{(123)} =&
    - 2 \tr(A_{j_1})\tr(A_{j_2})\tr(A_{j_3}) \tr(A_{j_4}) + 2 \bigg[ \tr(A_{j_3})\tr(A_{j_4}) \delta_{j_1j_2} + \tr(A_{j_2})\tr(A_{j_4}) \delta_{j_1j_3} + \tr(A_{j_1})\tr(A_{j_4}) \delta_{j_2j_3} \bigg] \\
    &- 4 \tr(A_{j_1}A_{j_2}A_{j_3}) \tr(A_{j_4}) +\tr(A_{j_1}[A_{j_2},A_{j_3}]A_{j_4})
    +\tr(A_{j_3}[A_{j_1},A_{j_2}]A_{j_4}) +\tr(A_{j_2}[A_{j_3},A_{j_1}]A_{j_4}).
\end{split}
\end{equation}
Going through all possible combinations for $r \leq 3$ leads to:
\begin{lemma}\label{condcoll}
    To show that the Hodge invariant $I_{14}$ cannot be obtained through a measurement of fourth moments using observables of Schmidt rank $r\leq3$, it is sufficient to show that the following terms always vanish for any orthonormal set $\{A,B,C\}$ of matrices:
    \begin{equation}\label{cond}
    \begin{split}
        -2\tr(A)^4+6 \tr(A)^2 - 4 \tr(A^3)\tr(A)&\stackrel{!}{=}0 \\
        -2\tr(A)^3\tr(B) + 6 \tr(A)\tr(B) - 4 \tr(A^3) \tr(B)&\stackrel{!}{=}0 \\
        -2\tr(A)^3 \tr(B) + 2 \tr(A) \tr(B) - 4 \tr(A^2B)\tr(A)&\stackrel{!}{=}0 \\
        -2\tr(A)^2\tr(B)^2 + 2 \tr(B)^2 - 4 \tr(A^2B) \tr(B)&\stackrel{!}{=}0\\
        -2\tr(A)^2\tr(B)\tr(C) - 4 \tr(ABC)\tr(A) + 2\tr(A^2[B,C])&\stackrel{!}{=}0\\
        -2\tr(A)^2\tr(B)\tr(C) + 2\tr(B)\tr(C)-4\tr(A^2B)\tr(C)&\stackrel{!}{=}0.\\ 
    \end{split}
    \end{equation}
\end{lemma}
These equations can be further simplified:
\begin{lemma}
    To show that the Hodge invariant $I_{14}$ cannot be obtained through a measurement of fourth moments using observables of Schmidt rank $r\leq3$, it is sufficient to show that for any orthonormal set $\{A,B,C\}$ of matrices
        \begin{equation}
        -2\tr(A)^2\tr(B)\tr(C) - 4 \tr(ABC)\tr(A) + 2\tr(A^2[B,C])\stackrel{!}{=} 0
    \end{equation}    
\end{lemma}
\begin{proof}
    The first and second equation in (\ref{cond}) factor into $-2\tr(A)^3+6\tr(A)-4\tr(A^3)$ and $\tr(A)$ (and $\tr(B)$, respectively). This equation is already known to always be zero by the discussion on the determinant (see the proof of Theorem \ref{r>2}).
    
    The third, fourth and sixth equation factor into $-2\tr(A)^2\tr(B) + 2\tr(B)-4\tr(A^2B)$ and $\tr(A)$ (and $\tr(B)$ and $\tr(C)$, respectively). This equation has also already been discussed in the proof of Theorem~\ref{r>2} when considering the determinant, and was also seen to be always $0$, leaving only one nontrivial condition.
\end{proof}

\begin{theorem}\label{hodge imp}
    Let $\mathcal{O} \in L(C^2)^{\otimes 2}_{\mathrm{herm.}}$ be an observable of operator Schmidt rank $\leq 4$, then $\mathcal{R}^{(4)}_{\mathcal{O}}$ does not depend on $I_{14}$. Thus, the Hodge invariant is of type 4 and therefore maximally difficult.
\end{theorem}
\begin{proof}
    Let $A_j = \frac{\tr(A_j)}{2} \mathbbm{1}+\frac{1}{\sqrt{2}}\overrightarrow{\alpha^{(j)}}\cdot \vec{\sigma}$ for $j=1,2,3$, as in the proof of Theorem~\ref{class0}. There, it was already shown that
    \begin{equation*}
        \tr(A_1A_2A_3) = - \frac{\tr(A_1)\tr(A_2)\tr(A_3)}{2} + \frac{i}{\sqrt{2}}\det(\overrightarrow{\alpha^{(1)}},\overrightarrow{\alpha^{(2)}},\overrightarrow{\alpha^{(3)}}).
    \end{equation*}
    Notice that
    \begin{equation*}
    \begin{split}
        [A_2,A_3] &= \left[\frac{\tr(A_2)}{2} \mathbbm{1}+\frac{1}{\sqrt{2}}\overrightarrow{\alpha^{(2)}}\cdot \vec{\sigma},\frac{\tr(A_3)}{2} \mathbbm{1}+\frac{1}{\sqrt{2}}\overrightarrow{\alpha^{(3)}}\cdot \vec{\sigma} \right] = \sum_{j,k=1}^3 \frac{\alpha^{(2)}_j\alpha^{(3)}_k}{2} [\sigma_j,\sigma_k] \\
        &= i \sum_{j,k,\ell=1}^3 \alpha^{(2)}_j\alpha^{(3)}_k \varepsilon_{jk\ell} \sigma_\ell = i \left(\overrightarrow{\alpha^{(2)}} \times \overrightarrow{\alpha^{(3)}}\right)\cdot \vec{\sigma}.
    \end{split}
    \end{equation*}
    Thus, 
    \begin{equation*}
    \begin{split}
        \tr(A_1^2[A_2,A_3]) =& \tr(\left(\frac{\tr(A_1)}{2} \mathbbm{1}+\frac{1}{\sqrt{2}}\overrightarrow{\alpha^{(1)}}\cdot \vec{\sigma}\right)^2\cdot i \left(\overrightarrow{\alpha^{(2)}} \times \overrightarrow{\alpha^{(3)}}\right)\cdot \vec{\sigma})\\
        =&2 \frac{\tr(A_1)}{2} \frac{1}{\sqrt{2}}i \sum_{j,k=1}^3 \alpha^{(1)}_j \left(\overrightarrow{\alpha^{(2)}} \times \overrightarrow{\alpha^{(3)}}\right)_k \tr(\sigma_j \sigma_k) +\frac{i}{\sqrt{2}^2} \sum_{j,k,\ell = 1}^3 \alpha^{(1)}_j\alpha^{(1)}_k\left(\overrightarrow{\alpha^{(2)}} \times \overrightarrow{\alpha^{(3)}}\right)_\ell \tr(\sigma_j \sigma_k \sigma_\ell) \\
        =&i \sqrt{2} \langle \overrightarrow{\alpha^{(3)}}, \left(\overrightarrow{\alpha^{(2)}} \times \overrightarrow{\alpha^{(3)}}\right) \rangle = \sqrt{2} i \det(\overrightarrow{\alpha^{(1)}},\overrightarrow{\alpha^{(2)}},\overrightarrow{\alpha^{(3)}}),
    \end{split}
    \end{equation*}
    where the second term vanishes, since $j\leftrightarrow k$ induces a minus, while at the same time not changing the sum.
    Putting all of this together leads to:
    \begin{equation*}
    \begin{split}
        &-2\tr(A_1)^2\tr(A_2)\tr(A_3)-4\tr(A_1A_2A_3)\tr(A_1)+2\tr(A^2_1[A_2,A_3]) \\
        =&-2\tr(A_1)^2\tr(A_2)\tr(A_3) - 4 \tr(A_1)\left( - \frac{\tr(A_1)\tr(A_2)\tr(A_3)}{2} + \frac{i}{\sqrt{2}}\det(\overrightarrow{\alpha^{(1)}},\overrightarrow{\alpha^{(2)}},\overrightarrow{\alpha^{(3)}}) \right) \\
        &+ 2 \sqrt{2} i \det(\overrightarrow{\alpha^{(1)}},\overrightarrow{\alpha^{(2)}},\overrightarrow{\alpha^{(3)}}) \\
        =&i\det(\overrightarrow{\alpha^{(1)}},\overrightarrow{\alpha^{(2)}},\overrightarrow{\alpha^{(3)}}) \underbrace{\left(\frac{-4}{\sqrt{2}} +2 \sqrt{2}\right)}_{=0} = 0.
    \end{split}
    \end{equation*}
    Thus, the criterion from the previous lemma has been fulfilled.
\end{proof}

\section{The determinant at $t=4$}\label{app:t4det}
While going against the initial classification of measurement difficulty, it may be of interest whether an $r=2$ measurement of $\det(T)$ is possible, if one considers fourth moments. The following discussion will show that this is not the case.

Denote with 
\begin{equation*}
    \tilde{S}_4 \coloneqq \{ (),(12),(13),(14),(23),(24),(34),(123),(1234),(1243),(1324),(1342),(1423),(1432)\}
\end{equation*}
the set of all relevant permutations for $d=2,t=4$ in the sense of Proposition \ref{sym}.

\begin{lemma}
    Let $d=2,t=4, r\leq 3$, and $\pi_A,\pi_B \in \tilde{S_4}$ then for an arbitrary bipartite qubit state $\varrho$ of Schmidt rank $r$ the quantity $\tr(\varrho^{\otimes 4} V_{\pi_A,\pi_B})$ can only depend on $\det(T)$ if $\pi_A$ is a 4-cycle, meaning ${\pi_A \in \{(1234),(1243),(1324),(1342),(1423),(1432)\}}$, and $\pi_B \in \{\pi_A, \pi_A^{-1}\}$.
\end{lemma}

\begin{proof}
    The only terms in the expansion of $\varrho^{\otimes 4}$ that can contribute to $\det(T)$ are those where $\mathbbm{1}\otimes \mathbbm{1}$ is chosen once, and $\sum\limits_{i,j}T_{ij} \sigma_i \otimes \sigma_j$ is chosen thrice. This especially means that no permutation with more than one fixpoint can appear, for the same reason as with the Hodge invariant. Furthermore, the proof of Theorem \ref{hodge imp} showed that, for $r\leq 3$, the cycle $(123)$ cannot contribute, and Lemma \ref{sym} states that all double-swaps $(ij)(k\ell)$ and all other 3-cycles can be removed using symmetry. This only leaves the 4-cycles. 
    \\
    
    Each $4$-cycle acts on $\sum\limits_{\substack{i_1,j_1,i_2, \\ j_2,i_3,j_3}} T_{i_1,j_1}T_{i_2,j_2}T_{i_3,j_3} \mathbbm{1}\otimes\mathbbm{1}\otimes \mathbbm{\sigma}_{i_1}\otimes \mathbbm{\sigma}_{j_1}\otimes \mathbbm{\sigma}_{i_2}\otimes \mathbbm{\sigma}_{j_2}\otimes \mathbbm{\sigma}_{i_3}\otimes \mathbbm{\sigma}_{j_3}$ as well as three other terms (which only differ by the position of $\mathbbm{1}\otimes \mathbbm{1}$). Acting with $V_{(1234),(1234)}$ on the first term and tracing leads to
    \begin{equation*}
        \sum\limits_{\substack{i_1,j_1,i_2, \\j_2,i_3,j_3}}T_{i_1,j_1}T_{i_2,j_2}T_{i_3,j_3} (2i)^2 \varepsilon_{i_1i_2i_3} \varepsilon_{j_1j_2j_3} = -4\cdot 3! \det(T).
    \end{equation*}
    Acting with any pair of permutations $V_{\pi_A,\pi_B}$ on any of the four terms and tracing leads to the above expression, but with possibly a different sign. Since four terms appear in total in each trace, differing signs might lead to cancellation.
    
    The possible sign combinations are listed in the following table:
    
    \begin{center}
        \begin{tabular}{ |c||c|c|c|c|c|c|   }
        \hline
          & $(1234)$ & $(1243)$ &$(1324)$ & $(1342)$ & $(1423)$ & $(1432)$\\
         \hline
         \hline
         1TTT&+&-&-&+&+&-\\
         T1TT&+&-&+&+&-&-\\
         TT1T&+&+&+&-&-&-\\
         TTT1&+&+&-&-&+&-\\
         \hline
        \end{tabular}  
        \end{center}
    This is to be read in the following way: In the expansion of $\tr(\varrho^{\otimes 4}V_{\pi_A,\pi_B})$, the prefactor in front of $-4\cdot 3! \det(T)$ can be obtained by summing the products of signs listed for $\pi_A$ and $\pi_B$ over all rows. For example, the prefactor of $-4\cdot3!\det(T)$ in $\tr(\varrho^{\otimes 4}V_{(1243)(1324)})$ is given by
    $(-1)\cdot(-1)+(-1)\cdot(+1)+(+1)\cdot(+1)+(+1)\cdot(-1)= 0$.
    In fact, for all choices of permutations the total prefactor turns out to only ever be non-zero if the signs either all agree (corresponding to $V_{\pi_A,\pi_A}$) or are all opposite to each other (corresponding to $V_{\pi_A,\pi_A^{-1}}$, since $(1234)^{-1} = (1432), (1243)^{-1} = (1342)$ and $(1324)^{-1} = (1423)$).
\end{proof}

\begin{lemma}\label{suf det}
    To show that the determinant $\det(T)=I_1$ cannot be obtained through a measurement of fourth moments using observables of operator Schmidt rank $r \leq 2$, it is sufficient to show that $x_\pi^{( \, \vec{j} \, )} = x_{\pi^{-1}}^{( \, \vec{j} \, )}$ for $\vec{j} \in \{1,2\}^4$, where $x_\pi^{( \, \vec{j} \, )}$ is defined via Proposition \ref{explicit} and $\pi\in \tilde{S}_4$ is an arbitrary 4-cycle.
\end{lemma}
\begin{proof}
    According to the proof of the last lemma, the total prefactor in front of $\det(T)$ is given by
    \begin{equation*}
        -4\cdot 3! \cdot \sum_{\substack{\pi \in \{(1234),(1243),\\(1324), (1342), \\(1423),(1432)\}}} \left(c_{\pi,\pi}-c_{\pi,\pi^{-1}}\right),
    \end{equation*}
    since, according to that proof, $c_{\pi,\pi}$ induces a "$+$" four times (terms with the same sign get multiplied), while $c_{\pi,\pi^{-1}}$ induces a "$-$" (terms with different signs get multiplied).
    Furthermore, by definition,
    \begin{equation*}
        c_{\pi,\pi} - c_{\pi,\pi^{-1}} = \sum_{j_1,j_2,j_3,j_4=1}^r s_{j_1} s_{j_2}s_{j_3} s_{j_4} x_\pi^{( \, \vec{j} \, )} \left(y_\pi^{( \, \vec{j} \, )} -  y_{\pi^{-1}}^{( \, \vec{j} \, )} \right).
    \end{equation*}
    Since the name of $x$ and $y$ is just a convention, the Lemma follows.
\end{proof}

Now everything is prepared for the proof of the analogon of Theorem \ref{r>2} for $t=4$.

\begin{theorem}\label{det t=4}
    Consider a bipartite state $\varrho$ in its Bloch representation (\ref{density matrix}). There is no observable $\mathcal{O} \in  L(\mathbb{C}^2)^{\otimes 2}_{\mathrm{herm.}}$ of operator Schmidt rank $r \leq 2$, such that $\mathcal{R}^{(t)}_{\mathcal{O}}(\varrho)$ depends on $\det(T)$, for $t\leq 4$.
\end{theorem}
\begin{proof}
    According to Theorem \ref{r>2}, only the case $t=4$ needs to be considered, which can be reduced to showing $x_\pi^{( \, \vec{j} \, )} - x_{\pi^{-1}}^{( \, \vec{j} \, )}=0$ for every 4-cycle $\pi \in S_4$ and $\vec{j} \in \{1,2\}^4$ according to Lemma \ref{suf det}.
    The first step is to find conditions similarly to those in Corollary \ref{condcoll}. Explicitly inverting the $14\times 14$ matrix $(M_2^{(4)})_{\pi,\pi' \in \tilde{S}_4}$ yields
    \begin{equation*}
    \begin{split}
        12\left(x_{(1234)}^{( \, \vec{j} \, )} - x_{(1234)^{-1}}^{( \, \vec{j} \, )}\right) =  &2        \tr(A_{j_1})\tr(A_{j_2})\tr(A_{j_3}) \tr(A_{j_4}) - 2 \big[ \tr(A_{j_3})\tr(A_{j_4}) \delta_{j_1j_2} + \tr(A_{j_2})\tr(A_{j_4}) \delta_{j_1j_3}\\
        &+ \tr(A_{j_1})\tr(A_{j_4}) \delta_{j_2j_3} \big] + 4 \tr(A_{j_1}A_{j_2}A_{j_3}) \tr(A_{j_4}) -2 \tr(A_{j_1}A_{j_2}A_{j_3}A_{j_4}) -\tr(A_{j_1}A_{j_2}A_{j_4}A_{j_3}) \\
        &+\tr(A_{j_1}A_{j_3}A_{j_2}A_{j_4}) +\tr(A_{j_1}A_{j_3}A_{j_4}A_{j_2})-\tr(A_{j_1}A_{j_4}A_{j_2}A_{j_3})+2\tr(A_{j_1}A_{j_4}A_{j_3}A_{j_2}).
    \end{split}
    \end{equation*}
    For the other $4$-cycles, the resulting equation is identical up to a relabeling.
    
    Note that Eq.~\eqref{cond} consists of a list of identities that were all shown to vanish. We use these identities in the remainder of this proof.
    
    For $r=1$ (or $\vec{j} = (1,1,1,1)$), all traces over the product of four matrices are $\tr(A^4_1)$, implying that the last $6$ vanish. The remaining terms equal $2\tr(A_1)^4 - 6 \tr(A_1)+4\tr(A_1)^3\tr(A_1)$, which is known from (\ref{cond}) to vanish.
    \\
    
    For $r=2$, consider first those $\vec{j}$ where one matrix ($\equiv A$) appears thrice, and the other one ($\equiv B$) once. Then the last six terms all equal $\tr(A^3B)$ and thus vanish. The first terms are then given by
    \begin{equation*}
    \begin{cases}
        2\tr(A)^3\tr(B)-6\tr(A)\tr(B)+4\tr(A^3)\tr(B) \quad & \text{if }\vec{j}=(1,1,1,2) \\
        2\tr(A)^3\tr(B)-2\tr(A)\tr(B)+4\tr(A^2B)\tr(A) \quad & \text{otherwise.}
    \end{cases}.
    \end{equation*} 
    Both of these terms are known to vanish from Eq.~\eqref{cond}.
    \\
    
    Consider now those $\vec{j}$ such that two matrices ($\equiv A,B$) appear twice each. The first $8$ terms are in this case given by
    \begin{equation*}
        2\tr(A)^2\tr(B)^2-2\tr(B)^2+4\tr(A^2B)\tr(A),
    \end{equation*}
    which is already known by be zero from (\ref{cond}). 
    
    The last $6$ terms are either given by $\tr(A^2B^2)$ or $\tr(ABAB)$ with vanishing prefactors, as for $\vec{j}=(1,1,2,2)$, they read 
    \begin{equation*}
        \tr(A^2B^2)\cdot (-10-5+5+10) + \tr(ABAB) \cdot (5-5) = 0,
    \end{equation*}
    and for $\vec{j} = (1,2,1,2)$ they read
    \begin{equation*}
        \tr(A^2B^2)\cdot (-5+5-5+5) + \tr(ABAB) \cdot (5-5) = 0
    \end{equation*}
    instead. Thus, all terms vanish, and the theorem follows.
\end{proof}

\section{The Kempe invariant}\label{app:kempe}
Analogously to the determinant, the first step in constructing a measurement process (or proving its impossibility) is to find permutations $\pi_A,\pi_B,\pi_C \in \tilde{S_3} \coloneqq \{(),(12),(13),(23),(123) \}$, such that $\tr(\varrho^{\otimes 3} V_{\pi_A,\pi_B,\pi_C})$ depends on
$\tr(T^{AB}T^{BC}T^{CA})$ for any tripartite qubit state $\varrho$ with Bloch decomposition as in Eq.~\eqref{bloch 3}.

Since $t=3$ (due to the invariant being of degree 3), this is only possible if one chooses the $T^{AB}$, the $T^{BC}$ and the $T^{CA}$ term once each in the expansion of $\varrho^{\otimes 3}$. Since the choice of the permutation $()$ traces out a full subsystem, any nontrivial objects can only appear for $\pi_A,\pi_B,\pi_C \in \{(12),(13),(23),(123) \}$.

One would now like to argue that, analogously to the determinant, each term including $(123)$ vanishes for $r\leq 2$, but there is an important distinction to make here. The concept of operator Schmidt decomposition only works for bipartite systems. In larger systems, each observable can still be written in the form $\mathcal{O} = \sum\limits_{j=1}^r A_j \otimes B_j \otimes C_j$, but this time $(A_j)_j$,$(B_j)_j$ and $(C_j)_j$ are no longer necessarily orthonormal. Luckily, the consideration from the discussion of the determinant can be generalized anyway:
\begin{lemma}
    Let $t=3$ and consider a three-partite qubit system. Then $x_{(123)}^{(\, \vec{j} \,)}$ (defined in Proposition \ref{explicit}) is always $0$ for $\vec{j} \in \{1,2\}^3$.
\end{lemma}

\begin{proof}
    Using (\ref{t=3 mat}) but adapting it to the current situation yields
    \begin{equation}\label{3x3new} \begin{bmatrix}
    x^{(\, \vec{j} \, )}_{()}\\
    x^{(\,\vec{j} \, )}_{(12)}\\
    x^{(\,\vec{j} \, )}_{(13)}\\
    x^{(\,\vec{j} \, )}_{(23)}\\
    x^{(\,\vec{j} \, )}_{(123)}\\
    \end{bmatrix}
= \frac{1}{12} \left(\begin{array} {cccccc} 2 & 0 & 0 & 0 & -2 \\[6pt]
    0 & 1 & -1 & -1 & 2 \\[6pt]
    0 & -1 & 1 & -1 & 2 \\[6pt]
    0 & -1 & -1 & 1 & 2 \\[6pt]
    -2 & 2 & 2 & 2 & -4
    \end{array} \right)\begin{bmatrix}
    \tr(A_{j_1})\tr(A_{j_2})\tr(A_{j_3})\\[6pt]
    \tr(A_{j_1}A_{j_2})\tr(A_{j_3})\\[6pt]
    \tr(A_{j_1}A_{j_3})\tr(A_{j_2})\\[6pt]
    \tr(A_{j_2}A_{j_3})\tr(A_{j_1})\\[6pt]
    \tr(A_{j_1}A_{j_2}A_{j_3})\\[6pt]
    \end{bmatrix}.
    \end{equation}

    For $\vec{j}=(1,1,1)$, this implies (written in terms of the eigenvalues of $A_1 \equiv A$):
    \begin{equation*}
    \begin{split}
        12x_{(123)}^{( \, \vec{j} \,)} &= -2 \tr(A)^3 +6 \tr(A^2)\tr(A) + \tr(A^3) = - 2(\lambda_1 + \lambda_2)^3 + 6 (\lambda_1^2 + \lambda_2^2)(\lambda_1+\lambda_2) -4 \lambda_1^3 + \lambda_2^3 \\
        &=(\lambda_1^3 + \lambda_2^3) (-2+6-4) + (\lambda_1^2 \lambda_2 + \lambda_2^2 \lambda_1) (-2\cdot 3 + 6) = 0.
    \end{split}
    \end{equation*}

    For $\vec{j} = (1,1,2)$, this implies (again written in terms of the eigenvalues of $A_1\equiv A$ and letting $A_2 \equiv B$):
    \begin{equation*}
    \begin{split}
        12x_{(123)}^{( \, \vec{j} \,)} &= -2 \tr(A)^2\tr(B) +2 \tr(A^2)\tr(B) + 4\tr(AB)\tr(A) - 4\tr(A^2B) \\
        &= -2 \tr(A)^2\tr(B) +2 \tr(A^2)\tr(B) + 4 (\lambda_1 B_{11} + \lambda_2 B_{22})(\lambda_1 + \lambda_2) - 4 (\lambda_1^2 B_{11} + \lambda_2^2 B_{22}) \\
        &= -2 \tr(A)^2\tr(B) +2 \tr(A^2)\tr(B) + B_{11}(4\lambda_1^2 +2 \lambda_1 \lambda_2 - 4 \lambda_1^2) + B_{22}(4\lambda_1 \lambda_2 +2\lambda_2^2 - 4 \lambda_2^2) \\
        &= 2\tr(B) \left[ -\tr(A)^2 + \tr(A^2)- 2 \lambda_1 \lambda_2 \right] = 2\tr(B) \left[-(\lambda_1+\lambda_2)^2  +\lambda_1^2 +\lambda_2^2 - 2\lambda_1\lambda_2\right] = 0.
    \end{split}
    \end{equation*}
\end{proof}

Thus, it only remains to discuss the choices of $V_{(12),(12),(12)}$ (and the other two terms of similar form), $V_{(12),(23),(31)}$ (and its five other permutations), as well as $V_{(12),(12),(13)}$ (and the 17 other terms of similar structure). They yield

\begin{equation*}
    \tr(\varrho^{\otimes 3} V_{(ij),(ij),(ij)}) = \frac{1}{8}\bigg[ 1 + |\vec{\alpha}|^2 + |\vec{\beta}|^2 + |\vec{\gamma}|^2  + \tr({(T^{AB})}^{\T}  T^{AB}) + \tr({(T^{BC})}^{\T}  T^{BC}) + \tr({(T^{CA})}^{\T}  T^{CA}) + ||W||_2^2\bigg],
\end{equation*}
where $||W||_2^2 = \sum\limits _{j,k,\ell = 1}^3 W_{jk\ell}^2$, as well as
\begin{equation*}
\begin{split}
    \tr(\varrho^{\otimes 3} V_{(12),(12),(13)}) &= \frac{1}{8}\bigg[1+ |\vec{\alpha}|^2+|\vec{\beta}|^2 + |\vec{\gamma}|^2 + \tr({(T^{AB})}^{\T}  T^{AB}) + \langle \alpha, T^{AC} \gamma \rangle + \langle \beta, T^{BC} \gamma \rangle +\sum_{jk\ell} W_{jk\ell} T_{jk}^{AB} \gamma_\ell\bigg],\\
    \tr(\varrho^{\otimes 3} V_{(12),(23),(13)}) &= \frac{1}{8} \bigg[1+ |\vec{\alpha}|^2+|\vec{\beta}|^2 + |\vec{\gamma}|^2 
    +  \langle \alpha, T^{AB} \beta \rangle  + \langle \alpha, T^{AC} \gamma \rangle + \langle \beta, T^{BC} \gamma \rangle) + \tr(T^{AB}T^{BC}T^{CA})\bigg].
\end{split}
\end{equation*}
The following lemma makes the situation more clear:

\begin{lemma} Let $i_1,i_2,i_3 \in \{1,2,3\}$ be pairwise distinct numbers and let $\varrho$ be an arbitrary tripartite qubit state written in terms of Eq.~\eqref{bloch 3}. Let $V_{\pi_A,\pi_B,\pi_C} = V_{\pi_A}\otimes V_{\pi_B} \otimes V_{\pi_C}$ for $\pi_A,\pi_B,\pi_C \in S_3$. Then
    \begin{equation}
    \begin{bmatrix*}[c]
        \tr(\varrho^{\otimes 3} V_{(i_1i_2),(i_1i_2),(i_1i_2)}) \\[3pt]
        \tr(\varrho^{\otimes 3} V_{(i_1i_2),(i_1i_2),(i_1i_3)}) \\[3pt]
        \tr(\varrho^{\otimes 3} V_{(i_1i_2),(i_1i_3),(i_1i_2)}) \\[3pt]
        \tr(\varrho^{\otimes 3} V_{(i_1i_3),(i_1i_2),(i_1i_2)}) \\[3pt]
        \tr(\varrho^{\otimes 3} V_{(i_1i_2),(i_2i_3),(i_3i_1)}) \\[3pt]
        \end{bmatrix*} = \frac{1}{8}\begin{bmatrix*}[c]
    ||W||_2^2\\
    \sum\limits_{jk\ell} W_{jk\ell} T^{AB}_{jk} \gamma_\ell\\
    \sum\limits_{jk\ell} W_{jk\ell} T^{CA}_{\ell j}\beta_k\\
    \sum\limits_{jk\ell} W_{jk\ell} \alpha_j T^{BC}_{k\ell}\\
    \tr(T^{AB}T^{BC}T^{CA}) \\
    \end{bmatrix*} + \text{\stackanchor{invariants available through}{direct measurements on subsystems.}}
    \end{equation}
    If $\mathcal{O}\in L(\mathbb{C}^2)^{\otimes 3}_{\mathrm{herm.}}$ is an observable of the form $\mathcal{O} = \sum\limits_{j=1}^r A_j \otimes B_j \otimes C_j$, where $r\leq 2$, then the listed invariants are the only nontrivial (i.e. not available through subsystem measurements) invariants that are possibly obtainable.
\end{lemma}

\begin{proposition}\label{Kempe Prop}
    There is no observable $\mathcal{O} = A \otimes B \otimes C \in L(\mathbb{C}^2)^{\otimes 3}_{\mathrm{herm.}}$ such that the invariant $\tr(T^{AB}T^{BC}T^{CA})$ is uniquely recoverable from a measurement of third moments.
\end{proposition}
\begin{proof}
    Notice that (\ref{3x3new}) implies that for $\vec{j}=(1,1,1)$:
    \begin{equation}
        x_{(12)}^{(1,1,1)} = x_{(13)}^{(1,1,1)} = x_{(23)}^{(1,1,1)} = 2 \tr(A)^3 - \tr(A^2) \tr(A).
    \end{equation}
    While this is not zero in general, this implies that the prefactor in front of $V_{\pi_A,\pi_B,\pi_C}$ is independent of the choice of 
    
    $\pi_A,\pi_B,\pi_C \in \{(12),(13),(23)\}$. Thus, the only nontrivial measurable invariant is given by
    \begin{equation}\label{Kempe r=1}
        3 ||W||_2^2 + 6 \sum_{jk\ell} W_{jk\ell} \left[ T^{AB}_{jk} \gamma_\ell + \alpha_j T^{BC}_{k\ell} + T^{CA}_{\ell j} \beta_k\right] + 6 \tr(T^{AB}T^{BC}T^{CA}).
    \end{equation}
\end{proof}

Notably, a measurement of (\ref{Kempe r=1}) is possible. For example, using $\mathcal{O}_1 = \left(\mathbbm{1} + \sigma_3\right)^{\otimes3}$ leads to 

$c_{\pi_A,\pi_B,\pi_C} = (2\tr(A)^3 - \tr(A^2)\tr(A))^3 = 512 \neq 0$.
\\

A plausible claim is now that it should be possible using multiple $r=2$ measurements to recover $\tr(T^{AB}T^{BC}T^{CA})$. Let 

\begin{equation*}
    \hat{c} \coloneqq     \begin{bmatrix*}[l]
        c_{(12),(12),(12)} + c_{(13),(13),(13)} + c_{(23),(23),(23)}\\
        c_{(12),(12),(13)} + c_{(12),(12),(23)} + c_{(13),(13),(12)} + c_{(13),(13),(23)} + c_{(23),(23),(12)} + c_{(23),(23),(13)}\\
        c_{(12),(13),(12)} + c_{(12),(23),(12)} + c_{(13),(12),(23)} + c_{(13),(23),(13)} + c_{(23),(12),(23)} + c_{(23),(13),(23)}\\
        c_{(13),(12),(12)} + c_{(23),(12),(12)} + c_{(12),(23),(23)} + c_{(23),(13),(13)} + c_{(12),(23),(23)} + c_{(13),(23),(23)}\\
        c_{(12),(23),(13)} + c_{(23),(13),(12)} + c_{(13),(12),(23)} + c_{(12),(13),(23)}+ c_{(13),(23),(12)}+ c_{(23),(12),(13)} \\
        \end{bmatrix*}.
\end{equation*}
Explicit calculation shows
\begin{equation}\label{strange type II}
    \hat{c}\left(\mathcal{O}_2 = \mathbbm{1}^{\otimes 3} + \sigma_3^{\otimes 3}\right) = \begin{bmatrix} 8/9 \\ 0 \\0\\0\\0 \end{bmatrix},
\end{equation}
allowing the measurement of $||W||_2^2$. Furthermore,
\begin{equation}
    \hat{c}\left(\mathcal{O}_3 = \mathbbm{1}^{\otimes 2} \otimes \left(\mathbbm{1} -\sigma_1 \right)+ \sigma_3^{\otimes 2} \otimes \left(\mathbbm{1} -\sigma_1 \right)\right) = \begin{bmatrix} 8/9 \\ 16/9 \\ 0 \\0\\0  \end{bmatrix},
\end{equation}
additionally allowing measurement of $\sum\limits_{jk\ell} W_{jk\ell} T_{jk}^{AB} \gamma_\ell$. Using the same observable (but permuting on which subsystem the $\mathbbm{1} - \sigma_1$ terms act) gives access to $\sum\limits_{jk\ell} W_{jk\ell} T_{\ell j}^{CA} \beta_k$ and $\sum\limits_{jk\ell} W_{jk\ell} \alpha_j T_{k\ell}^{BC}$ via symmetry. Thus, $\tr(T^{AB}T^{BC}T^{CA})$ can now be recovered from the possibility of measurement of $(\ref{Kempe r=1})$ and the four other measurements. This proves:

\begin{proposition}\label{kempe prop 2}
    Using the observables 
    \begin{align}\label{Kempe obs}
    \mathcal{O}_2&:=\mathbbm{1} \otimes \mathbbm{1}  \otimes \mathbbm{1}  + \sigma_3 \otimes \sigma_3 \otimes \sigma_3,\\
    \mathcal{O}_{3,a} &:= \left(\mathbbm{1} -\sigma_1 \right) \otimes \mathbbm{1} \otimes \mathbbm{1}+ \left(\mathbbm{1} -\sigma_1 \right) \otimes \sigma_3 \otimes \sigma_3,\\
    \mathcal{O}_{3,b} &:= \mathbbm{1} \otimes \left(\mathbbm{1} -\sigma_1 \right) \otimes \mathbbm{1}+ \sigma_3 \otimes \left(\mathbbm{1} -\sigma_1 \right) \otimes \sigma_3,\\
    \mathcal{O}_{3,c} &:= \mathbbm{1} \otimes \mathbbm{1} \otimes \left(\mathbbm{1} -\sigma_1 \right)+ \sigma_3 \otimes \sigma_3 \otimes \left(\mathbbm{1} -\sigma_1 \right)
    \end{align}
    (all fulfilling $r = 2$) in addition to some product observables on qubit subsystems, the invariant $I_{\mathrm{Kempe}}$ of a tripartite qubit system can be recovered via a measurement of third moments in the sense of Definition~\ref{def meas}:
    \end{proposition}

\end{document}